\documentclass[11pt]{article}
\usepackage[utf8]{inputenc}
\usepackage[boxed,linesnumbered,vlined]{algorithm2e}
\usepackage{palatino}
\usepackage{amssymb}
\usepackage{amsfonts}
\usepackage{amsmath}
\usepackage{amsthm}
\usepackage{latexsym}
\usepackage{color}
\usepackage{epsfig}
\usepackage{txfonts} 
\usepackage{hyperref}
\usepackage{bm}
\newcommand{\Es}{\mathsf{E}}
\newcommand{\opE}{\mathop{\mathbf{E}}}
\newcommand{\spann}{\mathrm{span}}
\usepackage{comment}
\usepackage{thmtools}
\usepackage{thm-restate}
\newcommand{\calA}{\mathcal{A}}

\newcommand{\ignore}[1]{}
\newenvironment{proofof}[1]{\par{\noindent \bf Proof of #1:}}{\qed\par}

\makeatletter
\renewcommand{\subsubsection}{\@startsection{subsubsection}{3}{0pt}{-12pt}{-5pt}{\normalsize\bf}}
\makeatother

\newcommand{\Net}{\mathsf{Net}}
\newcommand{\red}[1]{\color{red}{{#1}}}
\newtheorem{claim}{Claim}[section]

\newtheorem{proposition}[claim]{Proposition}

\newtheorem{lemma}[claim]{Lemma}
\newtheorem*{theorem*}{Theorem}
\newtheorem*{lemma*}{Lemma}
\newtheorem{theorem}[claim]{Theorem}
\newtheorem{definition}[claim]{Definition}
\newtheorem{corollary}[claim]{Corollary}

\newtheorem{fact}[claim]{Fact}
\newtheorem*{definition*}{Definition}

\def\L{\mathcal{L}}

\newcommand{\R}{{\mathbb{R}}}
\newcommand{\N}{{\bf N}}

\newcommand{\bx}{{\bf x}}

\newcommand{\by}{{\bf y}}
\newcommand{\bz}{{\bf z}}

\newcommand{\bZ}{{\bf Z}}

\newcommand{\calE}{\mathcal{E}}

\newcommand{\Ind}{\mathsf{Ind}}

\newcommand{\calC}{\mathcal{C}}

\newcommand{\E}{\mathop{\mathbf{E}}}

\newcommand{\eps}{\epsilon}

\newcommand{\poly}{\mathrm{poly}}

\newcommand{\fsm}{f_{\mathsf{sm}}}

\newcommand{\inr}[2]{\langle #1, #2 \rangle}

\newcommand{\bchi}{\boldsymbol{\chi}}

\newcommand{\Junta}{\mathcal{J}}

\newcommand{\normal}{\mathcal{N}}

\newcommand{\calG}{\mathcal{G}}
\DeclareMathOperator{\spn}{span}
\DeclareMathOperator{\Avg}{\mathcal{A}}
\DeclareMathOperator{\Lip}{Lip}

\newcommand{\Sm}{\mathsf{Sm}}

\begin{document}
\date{}
\title{Robust testing of low-dimensional functions}
\author{Anindya De\thanks{Supported by NSF grant CCF-1926872 and CCF-1910534}\\
University of Pennsylvania \\{\tt anindyad@cis.upenn.edu} \and Elchanan Mossel\thanks{Supported by Simons-NSF DMS-2031883, 
Vannevar Bush Faculty Fellowship ONR-N00014-20-1-2826, 
Simons Investigator award and 
NSF DMS-1737944} \\ MIT \\ {\tt elmos@mit.edu} \and Joe Neeman\thanks{Supported by the Alfred P.\ Sloan Foundation} \\ UT Austin \\ {\tt jneeman@math.utexas.edu}}
\maketitle

\begin{abstract}
A natural problem in high-dimensional inference is to decide if a classifier $f:\mathbb{R}^n  \rightarrow \{-1,1\}$ depends on a small number  of linear directions of its input data. Call a function $g: \mathbb{R}^n \rightarrow \{-1,1\}$, a  linear $k$-junta if it is completely determined by some $k$-dimensional subspace of the input space.  A recent work of the authors showed that linear $k$-juntas are testable. Thus there exists an algorithm to distinguish between: 
\begin{enumerate}
\item $f: \mathbb{R}^n  \rightarrow \{-1,1\}$ which is a linear $k$-junta with surface area $s$,  
\item $f$ is $\epsilon$-far from any linear $k$-junta with surface area $(1+\epsilon)s$,  
\end{enumerate}
where the query complexity of the algorithm is independent of the ambient dimension $n$. 


Following the surge of interest in noise-tolerant property testing, in this paper we prove a noise-tolerant (or robust) version of this result. 
Namely, we give an algorithm which given any $c>0$, $\epsilon>0$, distinguishes between 
\begin{enumerate}
\item  $f: \mathbb{R}^n  \rightarrow \{-1,1\}$ has correlation at least $c$ with some linear $k$-junta with surface area $s$. 
\item  $f$ has correlation at most $c-\epsilon$ with any linear $k$-junta with surface area at most $s$. 
\end{enumerate}
The query complexity of our tester is $k^{\mathsf{poly}(s/\epsilon)}$. 
Using our techniques, we also obtain a fully noise tolerant tester with the same query complexity for any class $\mathcal{C}$ of linear $k$-juntas with surface area bounded by $s$. As a consequence, we obtain a fully noise tolerant tester with query complexity $k^{O(\mathsf{poly}(\log k/\epsilon))}$ for the class of intersection of $k$-halfspaces (for constant $k$) over the Gaussian space. Our query complexity is independent of the ambient dimension $n$. 
Previously, no non-trivial noise tolerant testers were known even for a single halfspace.

\end{abstract}

\newpage
\section{Introduction}

To motivate our setting, consider the classical notion of a \emph{Boolean} junta:
a function $f: \{-1,1\}^n \rightarrow \{-1,1\}$ is said to be a $k$-junta if there are some $k$
coordinates $i_1, \ldots, i_k \in [n]$ such that $f(x)$ only depends on $x_{i_1}, \ldots, x_{i_k}$.
The fundamental results for testing juntas were obtained more than a decade ago; more
recently, spurred by motivation from several directions, several variants have appeared.
Most importantly for this work are the notions of \emph{tolerant testing}, in which we
estimate the distance to the class of juntas (as opposed to the usual testing, where we are
simply testing membership); and \emph{linear juntas}, a natural continuum generalization
of Boolean juntas.
In the current work, we combine these two perspectives and show that linear juntas are noise-tolerantly testable.


\subsection{Tolerant Junta Testing}
Recall that a property testing algorithm for a class of functions $\mathcal{C}$  is an algorithm which,
given oracle access to an $f:\{-1,1\}^n \rightarrow \{-1,1\}$ and a distance parameter $\epsilon>0$, satisfies
\begin{enumerate}
  \item If $f \in \mathcal{C}$, then the algorithm accepts with probability at least $2/3$;
  \item If $\mathsf{dist}(f,g) \ge \epsilon$ for every $g \in \mathcal{C}$, then the algorithm rejects with probability at least $2/3$. Here $\mathsf{dist}(f,g) = \Pr_{\bx \in \{-1,1\}^n} [f(\bx) \ne g (\bx)]$.
\end{enumerate}
The principal measure of the efficiency of the algorithm is its \emph{query complexity}. Also, the precise value of the confidence parameter is irrelevant and $2/3$ can be replaced by any constant $1/2 < c<1$. 

Fischer \emph{et~al.}\cite{FKRSS03} were
the first to study the problem of testing $k$-juntas
and showed that $k$-juntas can be tested with query complexity $\tilde{O}(k^2/\epsilon)$. The crucial feature of their algorithm is that the query complexity is independent of the ambient dimension $n$. Since then, there has been a long line of work on testing juntas~\cite{blais2009testing,blais2008improved,servedio2015adaptivity,Chen:2017:SQC,CLSSX18} and it continues to be of interest. The flagship result is that $k$-juntas can be tested with $\tilde{O}(k/\epsilon)$ queries and this is tight~\cite{blais2009testing,Chen:2017:SQC}. While the initial motivation to study this problem came from long-code testing~\cite{belgolsud98, PRS02} (related to PCPs and inapproximability), another strong motivation comes from the \emph{feature selection} problem in machine learning (see, e.g.~\cite{Blum:94, BlumLangley:97}).

\paragraph*{Tolerant testing} The definition of property testing above requires the algorithm to accept if and only if $f \in \mathcal{C}$. However, for many applications, it is important consider a \emph{noise-tolerant} definition of property testing. In particular, Parnas, Ron and Rubinfeld~\cite{parnas2006tolerant} introduced the following definition of noise tolerant testers.
\begin{definition}
  For constants $1/2>c_u> c_{\ell} \geq 0$ and a function class $\mathcal{C}$, a $(c_u,c_{\ell})$-noise tolerant tester  for $\mathcal{C}$ is an algorithm which given oracle access to a function $f: \{-1,1\}^n \rightarrow \{-1,1\}$
  \begin{enumerate}
    \item accepts with probability at least $2/3$ if
          $\min_{g \in \mathcal{C}} \mathsf{dist}(f,g) \le c_\ell$.
    \item  rejects with probability at least $2/3$ if
          $\min_{g \in \mathcal{C}} \mathsf{dist}(f,g) \ge c_u$.
  \end{enumerate}
  Further, a tester which is noise tolerant for any (given) $c_u>c_\ell \ge 0$ is said to be a ``\emph{fully noise tolerant}" tester.
\end{definition}
The restriction $c_u, c_\ell <1/2$ comes from the fact that most natural classes $\mathcal{C}$ are closed under complementation -- i.e., if $g \in \mathcal{C}$, then $-g \in \mathcal{C}$. For such a class $\mathcal{C}$ and for any $f$,  $\min_{g \in \mathcal{C}} \mathsf{dist}(f,g) \le 1/2$.
Further, note that the standard notion of property testing corresponds to a $(\epsilon,0)$-noise tolerant tester.

The problem of testing juntas becomes quite challenging in the presence of noise. Parnas \emph{et~al.}~\cite{parnas2006tolerant} observed that
any tester whose (individual) queries are uniformly distributed are inherently noise tolerant in a very weak sense. In particular, \cite{DLM+:07} used this observation to show that the junta tester of \cite{FKRSS03} is in fact a $(\epsilon, \mathsf{poly}(\epsilon/k))$-noise tolerant tester for $k$-juntas -- note that $c_\ell$ is quite small, namely $\mathsf{poly}(\epsilon/k)$. Later, Chakraborty~\emph{et al.}~\cite{chakraborty2012junto} showed that the tester of Blais~\cite{blais2009testing} yields a $(C\epsilon,\epsilon)$ tester (for some large but fixed $C>1$) with query complexity
$\mathsf{exp}(k/\epsilon)$. Recently, there has been a surge of interest in tolerant junta testing. On one hand, Levi and Waingarten showed that there are constants $1/2>\epsilon_1>\epsilon_2>0$ such that any non-adaptive $(\epsilon_1, \epsilon_2)$ tester requires $\tilde{\Omega}(k^2)$ non-adaptive queries. Contrast this with the result of Blais~\cite{blais2009testing} who showed that there is a non-adaptive tester for $k$-juntas with $O(k^{3/2})$ queries when there is no noise. In particular, this shows a gap between testing in the noisy and noiseless case.

In the opposite (i.e., algorithmic) direction a sequence of recent works improved on the results of \cite{chakraborty2012junto}. First, Blais~\emph{et~al.}~\cite{blais2018tolerant} improved on the results of \cite{chakraborty2012junto} by obtaining a small and explicit value of $C$.
Finally, De, Mossel and Neeman~\cite{de2019junta}
gave a fully noise tolerant tester for $k$-juntas on the Boolean cube with query complexity $O(2^k \cdot \mathsf{poly}(k/\epsilon))$.

\subsection{Linear Junta Testing}

In a recent work, De, Mossel and Neeman~\cite{DMN19a}  initiated the study of property testing of \emph{linear juntas}. A function $f: \mathbb{R}^n \rightarrow [-1,1]$ is said to be a \emph{linear $k$-junta} if   there are $k$ unit vectors $u_1, \ldots, u_k \in \mathbb{R}^n$ and $g: \mathbb{R}^k \rightarrow [-1,1]$ such that $f(x) = g(\langle u_1, x \rangle,  \ldots, \langle u_k, x \rangle)$.
In other words, $f$ is a linear $k$-junta if there is a subspace $E \coloneqq \spann (u_1, \ldots, u_k)$ of $\mathbb{R}^n$ such that $f(x)$ depends only on the projection of $x$ on the subspace $E$. The class of linear $k$-juntas is the $\mathbb{R}^n$-analogue of the class of $k$-juntas  on the Boolean cube


We note that the family of linear $k$-juntas includes important classes of functions that have been studied in the learning and testing literature. Notably it includes:
\begin{itemize}
  \item Boolean juntas: If $h : \{-1,1\}^n \to \{-1,1\}$ is a Boolean junta, then
        $f(x) : \mathbb{R}^n \to \{-1,1\}$ defined as $f(x) = h(\mathsf{sgn}(x_1),\ldots,\mathsf{sgn}(x_n))$ is a linear $k$-junta.
  \item {Functions of halfspaces: Linear $k$-juntas include as a special case both halfspaces and intersections of $k$-halfspaces. The testability of halfspaces was studied in ~\cite{MORS:09,MORS:10, ron2015exponentially}.}
\end{itemize}
The focus of the paper is on property testing of linear $k$-juntas. Observe that to formally define a testing algorithm, we need to define a notion of distance between functions $f$ and $g$ on $\mathbb{R}^n$.
In this work, we will use the $L^2(\gamma)$ metric, where $\gamma$ is the standard Gaussian measure.
That is, the distance between $f$ and $g$ is $(\E_{\bx \sim \gamma} [f(\bx) - g(\bx)^2])^{1/2}$.
Note that this reduces to $2\Pr_{\bx \sim \gamma} [f(\bx) \ne g(\bx)]$ when $f$
and $g$ are Boolean functions.
The choice of the standard Gaussian measure is well-established in the areas of learning and
testing~\cite{MORS:09, KNOW14, Neeman14, balcan2012active, chen2017sample, KOS:08, vempala2010learning, diakonikolas2018learning, harsha2012invariance}. It is particularly natural in our setup since the Gaussian measure is invariant under many linear transformations, e.g., all rotations.

De, Mossel and Neeman~\cite{DMN19a} obtained an algorithm for testing linear-$k$-juntas:
given query access to $f:\mathbb{R}^n \rightarrow \{-1,1\}$, it makes $\mathsf{poly}(k \cdot s/\epsilon)$ queries and distinguishes between
\begin{enumerate}
  \item $f$ is a linear-$k$-junta with surface area at most $s$ versus
  \item $f$ is $\epsilon$-far from any linear $k$-junta with surface area at most $s(1+\epsilon)$. 
\end{enumerate}
Here surface area of $f$ refers to the Gaussian surface area~\cite{Ledoux:94} of the set $f^{-1}(1)$~\cite{Ledoux:94}
-- see Definition~\ref{def:Gaussian-surface-area} for the precise definition.
Further, \cite{DMN19a} showed that a polynomial dependence on $s$ is necessary for any non-adaptive tester and consequently, an $\Omega(\log s)$ dependence is necessary for any tester\footnote{Recall that in a non-adaptive tester,  the query points are chosen independently of the target $f$.}. Informally,
without any smoothness assumption a linear junta (even a linear 1-junta on $\R^2$) can look arbitrarily
random to any finite number of queries.
Crucially,  \cite{DMN19a} achieves a query complexity which is independent of the ambient dimension $n$ -- thus, qualitatively matching the guarantee for junta testing on the Boolean cube.

\subsection{Our Results: Noise tolerant testing of linear-juntas}

In this paper, our focus is on the problem of noise tolerant testing of linear juntas.
The original motivation of~\cite{DMN19a} was for dimension reduction in statistical and ML models involving real valued data. Modern ML models are often overparametrized, but are nevertheless suspected
to output a predictor that is low-dimensional in some sense. The classical notion of juntas is not
appropriate for measuring dimensionality here, because there is  no natural choice of basis in many statistical models including PCA, ICA, kernel learning, or deep learning. 
This motivates the notion of a linear junta.
The problem of testing linear-juntas is thus closely related to the problem of {\em model compression} in machine learning, whose goal
is to take a complex predictor/classifier function and to output a simpler predictor/classifier
(see e.g.~\cite{bucilua2006model}).
Model compression is extensively studied in the context of deep nets, see e.g.,~\cite{ba2014deep},
and follow up work,
where the models are often rotationally invariant (with the caveat that the regularization often used in optimization might not be). Thus as a motivating example,~\cite{DMN19a} asked if given a complex deep net classifier, is there a classifier that has essentially the same performance and depends only on $k$ of the features? Observe that this is essentially the same question as asking whether the deep net classifier is a linear $k$-junta.

The main shortcoming of the motivation in~\cite{DMN19a} is that it is unrealistic to expect that in any of the statistical and ML models considered, the function constructed will be {\em exactly identical} to a function of a few linear direction. Rather, we only expect that the function will be {\em correlated} with a function of a few directions; this is the tolerant testing problem, and -- as evidenced by the long
history of tolerant testing in the Boolean case -- it is much more challenging.

The main result of this paper is a fully noise tolerant tester for $k$-linear juntas over the Gaussian space  whose query complexity is independent of the ambient dimension $n$. In particular, we prove the following:
\begin{restatable}{theorem}{testing}
  \label{thm:main-1}
  There is an algorithm \textsf{Robust-linear-junta-Boolean} which given parameters $1/2> c_u>c_\ell>0$, junta arity $k$ and surface area parameter $s$ and oracle access to $f: \mathbb{R}^n \rightarrow \{-1,1\}$  distinguishes between the following cases:
  \begin{enumerate}
    \item There is a linear-$k$-junta $g$ with surface area at most $s$ such that $\mathsf{dist}(f,g) \le c_\ell$.
    \item For all linear-$k$-juntas $g$ with surface area at most $s$, $\mathsf{dist}(f,g) \ge c_u$.
  \end{enumerate}
  The query complexity of the tester is $k^{\mathsf{poly}(s/\epsilon)}$ where $\epsilon = c_u - c_\ell$ and tester makes non-adaptive queries.
\end{restatable}
Note that qualitatively this result implies the main result of \cite{DMN19a} -- thus a dependence on $s$ is necessary, though presumably, one can achieve a polynomial dependence on $s$. In fact, the result here is \emph{qualitatively stronger} than \cite{DMN19a} as our ``soundness guarantee" does not require relaxing the surface area to $s(1+\epsilon)$. On the other hand, 
the query complexity here as an exponential dependence on $s$ vis-a-vis \cite{DMN19a} which has a polynomial query complexity in all the parameters.


It is not hard to see that tolerant testing is essentially equivalent to estimating the maximum
correlation between a function and a class.
In particular, Theorem~\ref{thm:main-1} follows from the following result about estimating correlation.
Here (and in most of this work), it is more convenient to consider functions with values in $[-1, 1]$. For
these functions, we need a more general notion of smoothness: we will define the notion of
\emph{$s$-smooth functions} later (in Definition~\ref{def:s-smooth}); for now, we just note
that it includes both Lipschitz functions and Boolean functions with bounded surface area.
\begin{restatable}{theorem}{testingsmooth}
  \label{thm:main-2}
  There is an algorithm \textsf{Correlation-smooth-junta} which, given parameters $\epsilon>0$, junta arity $k$ and smoothness parameter $s$ and oracle access to $f: \mathbb{R}^n \rightarrow [-1,1]$, outputs
  an estimate $\hat{\rho}_{\mathbb{R}^n, k, s}(f)$ with the following guarantee:
  \[
    \big|\hat{\rho}_{\mathbb{R}^n, k, s}(f)- {\rho}_{\mathbb{R}^n, k, s}(f) \big| \le \epsilon.
  \]
  Here ${\rho}_{\mathbb{R}^n, k, s}(f)$ is the maximum  correlation of $f$ with any $s$-smooth $k$-linear junta. The query complexity of the algorithm is $k^{\mathsf{poly}(s/\epsilon)}$.
\end{restatable}
In particular, Theorem~\ref{thm:main-1} follows as a simple corollary of Theorem~\ref{thm:main-2}. 

\subsubsection{List decoding the linear-invariant structure.}
Given the previous theorem it is natural to ask for more, i.e., not just test if the function is a linear-junta but also find a junta in number of queries that depends only on $k$ and $s$ (but not on $n$) that has almost maximal correlation with $f$.
In other words, the goal is to find, with query complexity independent of $n$, a function $g: \mathbb{R}^k \rightarrow \{-1,1\}$ such that there exists a {projection} matrix $A: \mathbb{R}^n \rightarrow \mathbb{R}^k$ and such that the correlation between $f$ and
$g(A x)$ is at least ${\rho}_{\mathbb{R}^n, k, s}(f) - \epsilon$.

In the case where $f$ is a linear $k$-Junta with bounded surface area, i.e., ${\rho}_{\mathbb{R}^n, k, s}(f) = 1$, \cite{DMN19a} provided such an algorithm with query complexity that is exponential in $k$. In the noisy case, we could have multiple different Juntas that have optimal or close to optimal correlation with $f$. 
Ideally we would like to find all those functions, which can be thought of as ``list decoding" the Juntas that are hidden in $f$.

There is some subtlety in the meaning of ``all'' here; for example, if $f$ is
a linear 1-Junta with some added noise and we set $k=2$, then there can be
a huge number (i.e. growing quickly with $n$) of linear 2-Juntas that are
highly correlated with $f$, just because there is a lot of flexibility
in choosing the second direction and defining the function in that direction.
For this reason, rather than identifying all highly-correlated linear Juntas,
we only identify their averages on a set of interesting directions; for
a subspace $E$ of $\R^n$ and a function $g: \R^n \to \R$, let $\calA_E g$ 
be obtained from $g$ by averaging over the  directions
orthogonal to $E$ (see Definition~\ref{def:AE} for a full definition).

\begin{restatable}{theorem}{learnall}
  \label{thm:main-4}
  There is an algorithm \textsf{Learn-all-invariant-structures} which, given
  parameters $\rho, \epsilon>0$, junta arity $k$,  smoothness parameter $s$ and
  oracle access to $f: \mathbb{R}^n \rightarrow [-1,1]$, outputs a set $\calG$
  of functions $\mathbb{R}^k \rightarrow [-1,1]$ so that the following hold:
  \begin{itemize}
      \item
          for every $\hat g \in \calG$ there exists
  an orthonormal set of vectors $w_1, \ldots, w_k \in \mathbb{R}^n$ such that
  \[
     \big| \mathbf{E}[f(x) \hat g(\langle w_1, x\rangle, \ldots, \langle w_k, x \rangle)] -\rho| = O(\epsilon),
  \]
  and
    \item
   for
  every linear $k$-Junta $g: \R^n \to [-1, 1]$ with
  $\big|\mathbf{E}[f(x) g(x)] -\rho\big| \le \epsilon$,
  there exists a function $\hat g \in \calG$ and
  an orthonormal set of vectors $w_1, \ldots, w_k \in \mathbb{R}^n$ such that,
        with $E = \spn\{w_1, \dots, w_k\}$, we have
        \[
            \mathbf{E}[\big((\Avg_E g)(x) - \hat g(\langle w_1, x\rangle, \ldots, \langle w_k, x \rangle)\big)^2]
            \le O(\epsilon).
        \]
        Additionally, 
        \[
 \mathbf{E} [f(x) g(x)] \approx_{O(\epsilon)} \mathbf{E} [f(x) (\Avg_E g)(x)]. 
 \]
  \end{itemize}
  The query complexity of the algorithm is $k^{\poly(s/\epsilon)}$.
\end{restatable}

Informally, the theorem states that it is possible to find the ``linear-invariant" structures (i.e., the structure up to unitary transformation) of all Juntas
that are almost optimally correlated with
$f$ in number of queries that depends on $s$ and $k$.  We note that one cannot hope to output the relevant directions $w_1, \ldots, w_k$ explicitly as even describing these directions will require $\omega(n)$ bits of information and thus, at least those many queries.

The significance of Theorem~\ref{thm:main-4} is related to one of the main
difficulties in tolerant testing: there can be a large number of linear Juntas
having almost optimal correlation with $f$. This in in contrast with the usual
testing problem, because if $f$ is in fact a linear $k$-Junta then there is (obviously)
only one linear $k$-Junta that is equal to $f$.

Even in the noiseless case, Theorem~\ref{thm:main-4} improves on the the results of 
\cite{DMN19a} which provided an algorithm for learning the linear structure with query complexity that is exponential in $k$. 
We note that in~{\cite{DMN19a} it was incorrectly stated (without proof) that the exponential dependence on $k$ is necessary.

Thanks to Theorem~\ref{thm:main-4}, we are also able to tolerantly test certain
\emph{subclasses} of linear Juntas; this is significant because in general
the testability of a class does not imply the testability of a subclass.

\begin{definition}~\label{def:induced-class}
    Let $\calC$ be any collection of functions mapping $\R^k$ to $\{-1, 1\}$.
    For any $n \in \N$, define the \emph{induced class of $\calC$} by
    \[
        \Ind(\calC)_n = \{f: \exists g \in \calC \emph{ and orthonormal vectors $w_1, \dots, w_k$
        such that } f(x) = g(\inr{w_1}{x}, \dots, \inr{w_k}{x})\}.
    \]

\end{definition}

Note that every $f \in \Ind(\calC)_n$ is a linear $k$-Junta. As an example, if $\calC$ is the class of intersections of $k$-halfspaces over $\mathbb{R}^k$, then $\Ind(\calC)_n$ is the class of intersections of $k$-halfspaces over $\mathbb{R}^n$. 

 
\begin{restatable}{theorem}{ctest}~\label{thm:C_test}
    Let $\calC$ be a collection of functions mapping $\R^k$ to $[-1, 1]$ such that each $f \in \mathcal{C}$ is $s$-smooth. 
  There is an algorithm \textsf{Robust-$\calC$-test} which given parameters $1/2> c_u>c_\ell>0$, junta arity $k$, surface area parameter $s$, and oracle access to $f: \mathbb{R}^n \rightarrow \{-1,1\}$, distinguishes between the following cases:
  \begin{enumerate}
      \item There is a linear-$k$-junta $g \in \Ind(\calC)_n$ with surface area at most $s$ such that $\mathsf{dist}(f,g) \le c_\ell$.
      \item For all linear-$k$-juntas $g \in \Ind(\calC)_n$, $\mathsf{dist}(f,g) \ge c_u$.
  \end{enumerate}
  The query complexity of the tester is $k^{\mathsf{poly}(s/\epsilon)}$ where
  $\epsilon = c_u - c_\ell$, and the tester makes non-adaptive queries.
\end{restatable}

As an immediate corollary, this implies that there is a fully noise tolerant tester for intersections of $k$-halfspaces with query complexity $k^{\mathsf{poly}(\log k/\epsilon)}$. Previously, no noise tolerant tester was known for even a single halfspace~\cite{MORS:10}.

\subsection{Techniques}\label{sec:techniques}
For ease of exposition, here we just explain the technique for proving Theorem~\ref{thm:main-2}. The high level proof technique for the other results is essentially the same albeit sometimes with added technical complications. 
The techniques of the current paper build on those of \cite{DMN19a}. We briefly recap
the main ideas of~\cite{DMN19a}, restricted for now to the non-tolerant setting:
\begin{enumerate}
  \item[I.] If we sample $T = \mathsf{poly}(k/\epsilon)$ random points $\bx_1, \ldots, \bx_T$  from the standard Gaussian measure $\gamma_n$ and let $E = \spann (\nabla f(\bx_1), \ldots, \nabla f(\bx_T))$, then if $f$ is a linear $k$-Junta then with high probability,  $f$ has correlation $1-\epsilon$ with some linear $k$-junta defined on the space $E$.
  \item[II.]  For each $x,z$, it is possible to estimate, in number of samples polynomial in $k$, quantities  such as $\langle z, \nabla f(x) \rangle$ and $\langle \nabla f(x_i) , \nabla f(x_j) \rangle$ up to small error. Thus,  for a randomly chosen $\bz \sim \gamma_n$, we can (implicitly, in a sense
        to be made precise later) compute the orthogonal projection of $\bz$ on $E$. (Note that a naive estimation of $\nabla f(\bx)$, or even $\langle \nabla f(\bx_1), \nabla f(\bx_2)\rangle$, requires a number of samples that depends on $n$.)
\end{enumerate}

Observe that the implicit projection allows \cite{DMN19a} to effectively reduce the dimension of the ambient space to $T=\mathsf{poly}(k/\epsilon)$, which is independent of $n$.
We then take an $\epsilon$-net of linear $k$-juntas over $E$ with surface area $s$.  The size of this net depends only on $s$, $k$ and $\epsilon$. For each function in the net, one can estimate its distance to $f$; by
iterating over all functions in the net, one can check if $f$ is close to a linear $k$-junta.
This last step is different from the one in~\cite{DMN19a}, and in fact it is slightly worse.
By following the ideas in the current paper, one can show that it gives a tester for
with query complexity of $k^{O(s^2/\epsilon^2)}$. The advantage of the modification, however,
is that it yields a method that is more robust to noise.


\subsubsection*{Adding tolerance.}
In adapting the outline above to the setting of tolerant testing, the
main challenge is to imitate step I above.
Our main structural result roughly shows that if $f$ has correlation $c$
with some linear $k$-junta of surface area at most $s$, then with high probability $f$ is at least $c-\eps$ correlated with an $s$-smooth linear $k$-junta defined on $E$.
In fact, we need to define $E$ more carefully than what is outlined above, and a good error
analysis is crucial. If we were to combine our new structural result with a naive error analysis, it would
give a query complexity that is exponential in $\poly(k)$.

The proof of our structural result is non-trivial. At the intuitive level it is related to the idea of using SVD for PCA. In our case, we have a function, rather than a collection of data, and the right geometric information is encoded by gradients (of a smoothed version of this function). The procedure of using SVD to extract informative directions from the data can be thought of as ``{\em Gradient Based} PCA". The proof that this procedure actually extracts the relevant dimensions requires combining linear algebraic and Poincare style geometric estimates in just the right way.
Another challenge comes from the fact that gradients are only approximate due to sampling effect. We use results from random matrix theory to control the effect of sampling.

The methodology of Gradient Based PCA also allows us to improve on the results of  \cite{DMN19a} in the noiseless case for finding an approximation of the Junta. This has to do with the fact that the results of \cite{DMN19a} used a more naive Gram–Schmidt based process to extract the linear structure which resulted in exponential query complexity, compared with the polynomial query complexity we achieve in the current work.



Another key new ingredient of this current work is the net argument outlined above.
We show that
the class of $s$-smooth linear $k$ juntas has an $\epsilon$-net of size
$\exp \exp ((s^2 \log k)/\epsilon^2)$. For each function in the net, we can use the implicit projection algorithm to compute the correlation between this function and $f$ up to error $\epsilon$.
The maximum of these correlations gives a good estimate of the best correlation
between $f$ and any linear $k$-junta with surface area $s$.
This concludes the proof sketch of Theorem~\ref{thm:main-2}.

We note that there is a high-level similarity between the current proof and the proof that Boolean juntas are tolerant testable~\cite{de2019junta}.
Both strategies are based on oracle access to influential ``directions" followed by a search
for juntas depending only on those influential directions.
In the Boolean case, the ``directions" are influential variables, while here the directions are  given by gradients of the function $f$.
Note, however, that the Boolean case is easier, since the coordinates on the Boolean cube are automatically orthogonal, while in the continuous setup, ``relevant directions" as sampled from data are often not orthogonal and indeed can be close to parallel. This is one of the major reasons we needed to introduce and analyze the methodology of gradient based PCA.

\paragraph{Related work:} Besides being related to the long line of work on junta testing, the current work is also connected to the rich area of learning and testing of threshold functions. In particular, an immediate corollary of Theorem~\ref{thm:C_test} gives a fully noise tolerant tester for any function of $k$-halfspaces over the Gaussian space. Despite prior work on testing of halfspaces~\cite{MORS:10, MORS:09random, ron2015exponentially}, until this work, no non-trivial noise tolerant tester was known even for a single halfspace. 

Finally, we remark that  the notion of noise in property testing (including this paper) is the so-called \emph{adversarial label noise}~\cite{KSS:94}. This  is   stronger than many other noise models in literature such as the \emph{random classification noise}~\cite{KSS:94} and \emph{Massart noise}~\cite{massart2006risk}. Both these models are important from the point of view of learning theory -- in particular, halfspaces (and polynomial threshold functions) are known to be efficiently learnable~\cite{BFK+:97, diakonikolas2019distribution} in both these models of noise even when the background distribution is arbitrary. On the other hand, for arbitrary background distributions, halfspaces are hard to learn in the adversarial label noise model~\cite{GR:06, daniely2016complexity}. In contrast, the tester in the current paper is in the adversarial label noise model but works only when the background distribution is the Gaussian. This discussion raises the intriguing possibility that halfspaces (and more generally linear juntas) can be tested in the \emph{distribution free model}~\cite{HalevyKushilevitz:07} with weaker models of noise such as the Massart noise.

\section{Preliminaries}
In this section, we list some useful definitions and
technical preliminaries. We begin with some definitions and properties of projections and averages.

\begin{definition}~\label{def:AE}
    For a subspace $E$ of $\mathbb{R}^n$, we denote by $\Pi_E: \mathbb{R}^n \rightarrow \mathbb{R}^n$
    the orthogonal projection onto $E$.
    For a subspace $E$ and $f: \mathbb{R}^n \rightarrow \mathbb{R}$, we define the operator $\calA_E$ as
    $\calA_Ef(x) = \opE_{\bz \sim \gamma_n} [f(\Pi_E x + \Pi_{E^{\perp}} \bz)]$,
    where $\gamma_n$ is the standard $n$-dimensional Gaussian measure.

    Finally, for any subspace $E$, we define $\Junta_E = \{f :$  if  $\Pi_E z = \Pi_E x$, then $f(x) = f(z)\}$.
\end{definition}

One way to understand the operator $\mathcal{A}_E$ is that it averages $f$  on the directions orthogonal to the subspace $E$.
The next lemma lists some useful properties of the operators $\Pi_E$ and
$\calA_E$. Let $C_b^1(\mathbb{R}^n)$ be the class of differentiable functions
$f$ such that $f(x)$ and $\nabla f(x)$ are bounded.
\begin{lemma}\label{lem:avg-props}
    For any $f \in C_b^1(\mathbb{R}^n)$, any subspaces $E \subset E' \subset \mathbb{R}^n$, and any $x \in \mathbb{R}^n$,
    the following hold:
    \begin{enumerate}
        \item \label{it:depends-on-E}
              If $\Pi_E z = \Pi_E x$ then $\Avg_E f(x) = \Avg_E f(z)$.
              In other words, $\Avg_E f \in \Junta_E$.
        \item \label{it:gradient-of-average}
              $ \displaystyle
                  (\nabla \Avg_E f)(x) = \E_{\bz}[\Pi_E \nabla f(\Pi_E x + \Pi_{E^\perp} \bz)]
              $
        \item \label{it:average-tower}
              $ \displaystyle
                  (\Avg_E \Avg_{E'} f)(x) = (\Avg_E f)(x)
              $
        \item \label{it:average-inner-product}
              For all $g \in \Junta_E$,
              $\displaystyle
                  \E_{\bx} [g (\bx) \Avg_E f(\bx)] = \E_{\bx}[g(\bx) f(\bx)]
              $
        \item \label{it:self-adjoint}
              For all $g \in L^2(\gamma)$,
              $\displaystyle
                  \E_{\bx} [(\Avg_E f) (\bx) g(\bx)] = \E_{\bx}[f (\bx)(\Avg_E g)(\bx)].
              $
    \end{enumerate}
\end{lemma}

Note that parts~\ref{it:depends-on-E} and~\ref{it:average-inner-product}
can be interpreted as saying that $\Avg_E f$ is the orthogonal projection
(in the $L^2(\gamma)$ sense) of $f$ onto $\Junta_E$.

\begin{proof}
    ~\\
    \begin{enumerate}
        \item Part~\ref{it:depends-on-E} is immediate from the definition of $\Avg_E$.

        \item To prove part~\ref{it:gradient-of-average},
              fix $v \in \mathbb{R}^n$. Then
              \[
                  (\Avg_E f)(x) - (\Avg_E f)(x - v) =
                  \E_{\bz}[ f(\Pi_E x + \Pi_{E^\perp} \bz) - f(\Pi_E x - \Pi_E v + \Pi_{E^\perp} \bz)]
              \]
              Replacing $v$ by $hv$ and sending $h \to 0$, we obtain (and there is no
              trouble exchanging the limit and the expectation, because $f$ is Lipschitz)
              \[
                  (\nabla_v \Avg_E f)(x) = \E_{\bz} [\nabla_{\Pi_E v} f (\Pi_E x + \Pi_{E^\perp} \bz)].
              \]
              This proves the second item.

        \item  Part~\ref{it:average-tower} follows from the fact that if $E \subset E'$
              then $\Pi_E \Pi_{E'} z = \Pi_E z$ and $\Pi_{E} \Pi_{(E')^\perp} z = 0$ for every $z$.
              Indeed, if $\bz$ and $\bz'$ are independent standard Gaussian variables then
              \[
                  (\Avg_{E} \Avg_{E'} f)(x)
                  = \E_{\bz, \bz'}[f(\Pi_E(\Pi_{E'} x + \Pi_{(E')^\perp} \bz') + \Pi_{E^\perp} \bz)]
                  = \E[f(\Pi_E x + \Pi_{E^\perp} \bz)]=(\Avg_{E} f)(x).
              \]

        \item   For Item~\ref{it:average-inner-product}, let $\bz$ and $\bz'$ be
              standard Gaussian variables. Since $g \in \Junta_E$, we have
              $g(z) = g(\Pi_E z + \Pi_{E^\perp } z')$. Hence,
              \[
                  \E[g \Avg_E f]
                  = \E_{\bz,\bz'}[g(\bz) f(\Pi_E\bz + \Pi_{E^\perp} \bz')]
                  = \E[g(\Pi_E \bz + \Pi_{E^\perp} \bz') f(\Pi_E \bz + \Pi_{E^\perp} \bz')].
              \]
              Since $\Pi_E \bz + \Pi_{E^\perp} \bz'$ has the same distribution as $\bz$,
              the claim follows.

        \item    Item~\ref{it:self-adjoint} follows from applying claim~\ref{it:average-inner-product} twice:
              \[
                  \E_{\bx} [(\Avg_E f)(\bx) g(\bx)] = \E [(\Avg_E f)(\bx) (\Avg_E g)(\bx)] = \E[f(\bx) (\Avg_Eg)(\bx)].
                  \qedhere
              \]
    \end{enumerate}
\end{proof}

As an immediate consequence of the properties above, we have the following two basic properties
of $\nabla \Avg_E f$:

\begin{claim}~\label{claim:gradient-AE}
    \begin{enumerate}
        \item~\label{it:average-gradient-AE} $
                  \opE_{\bx} [\nabla \Avg_E f(\bx)] = \opE_{\bx} [\Pi_E \nabla f(\bx)]
              $.
        \item~\label{it:variancee-gradient-AE}  $
                  \opE_{\bx} [\Vert \nabla \Avg_E f(\bx) \Vert_2^2] \le \opE_{\bx} [\Vert\Pi_E \nabla f(\bx) \Vert_2^2].
              $
    \end{enumerate}
\end{claim}
\begin{proof}
    Item~\ref{it:average-gradient-AE} follows by
    averaging over Item~\ref{it:gradient-of-average} from Lemma~\ref{lem:avg-props}. To get Item~\ref{it:variancee-gradient-AE},
    first observe that by Jensen's inequality (applied on Item~\ref{it:gradient-of-average} from Lemma~\ref{lem:avg-props}), we have
    \[
        \Vert (\nabla \Avg_E f)(x) \Vert^2 \leq \E_{\bz}[\Vert \Pi_E \nabla f(\Pi_E x + \Pi_{E^\perp} \bz )\Vert^2].
    \]
    Averaging over $\bx \sim \gamma$ and observing that the distribution of $\Pi_E \bx + \Pi_{E^\perp} \bz$ is the same as that of $\bx$, we have Item~\ref{it:variancee-gradient-AE}.
\end{proof}

\subsection{Smoothness and Juntas}

The notion of smoothness that we will use in this work depends on the notion of Gaussian noise.
In particular, we use the following Gaussian noise operator:
\begin{definition}~\label{def:Ornstein}
    For $t \ge 0$ and $f \in L_2(\gamma_n)$, we define $P_tf:\mathbb{R}^n \rightarrow \mathbb{R}$ as
    \[
        P_tf(x) = \mathbf{E}_{\by} [f(e^{-t}x + \sqrt{1-e^{-t}} \by)].
    \]
    The operator $P_t$ forms a semigroup, i.e., $P_{t} P_{t'} f = P_{t+t'}f$. Further, for $t>0$, $P_t f$ is infinitely differentiable.
\end{definition}
We recall here a basic property of this noise operator -- namely, that $P_t$ makes any bounded function Lipschitz, a fact that can be derived,
for example, from (2.3) in~\cite{Ledoux:00}.
\begin{fact}\label{fact:P_t-lipschitz}
    For any $f: \R^n \to [-1, 1]$ and any $t > 0$, $P_t f$ is $\frac{C}{\sqrt t}$-Lipschitz for an absolute constant $C$.
\end{fact}

Now we come to the notion of \emph{$s$-smooth} functions:
\begin{definition}~\label{def:s-smooth}
    A function $f: \mathbb{R}^n \rightarrow [-1,1]$ is referred to as \emph{$s$-smooth} if
    for all $t>0$, $$\mathbf{E}[| f(\bx) - P_t  f(\bx) | ] \le s \sqrt{t}.$$
    In this case, we say that $\Sm(f) \le s$.
\end{definition}

To help illustrate the definition, let us recall the notion of Gaussian surface area:
\begin{definition}~\label{def:Gaussian-surface-area}
    For a Borel set $A \subseteq \mathbb{R}^n$, we define its Gaussian surface area, denoted by $\Gamma(A)$ to be
    \[
        \Gamma(A) = \liminf_{\delta \rightarrow 0} \frac{\mathsf{vol}(A_{\delta} \setminus A)}{\delta}.
    \]
    Here $A_{\delta}$ denotes the set of points which are at Euclidean distance at most $\delta$ the set $A$.
\end{definition}

The next proposition shows that the class of $s$-smooth functions functions of bounded surface area.
Later, we will also show that the notion of $s$-smoothness is equivalent to a certain decay in the
Hermite coefficients (which can also be used to show that if $\Sm(f) \le C\E[\|\nabla f\|^2]$,
so for example Lipschitz functions are $s$-smooth).
\begin{proposition}~\label{prop:smooth-inclusion}
    \begin{enumerate}
        \item~\label{it:prop-1} If $f:\mathbb{R}^n \rightarrow \{-1,1\}$ with surface area at most $\frac{s \cdot \sqrt{\pi}}{2}$, then $\Sm(f) \le s$.
        \item~\label{it:prop-2} Let $E$ be any subspace of $\mathbb{R}^n$. If $f$ is $s$-smooth, then so is $\Avg_E f$.
    \end{enumerate}
\end{proposition}
\begin{proof}
    \begin{enumerate}
        \item Part~\ref{it:prop-1} was proved by Pisier~\cite{Pisier:86} and Ledoux~\cite{Ledoux:94}.
        \item To prove Part~\ref{it:prop-2}, observe that the operators $\Avg_E$ and $P_t$ commute. Thus,
              \begin{align*}
                  \mathbf{E}[| \Avg_E f(\bx) - P_t   \Avg_E f(\bx) | ] & = \mathbf{E}[| \Avg_E f(\bx) -    \Avg_E P_t f(\bx) | ].
              \end{align*}
              However, Jensen's inequality implies that for any $f$ and $g$,
              $\mathbf{E}_{\bx} [| \Avg_E f(\bx) -  \Avg_E g(\bx) |] \le \mathbf{E}_{\bx} [|  f(\bx) -   g(\bx) |].$ This finishes the proof.
    \end{enumerate}

\end{proof}

We now define the class of $s$-smooth linear $k$-juntas.

\begin{definition}~\label{def:junta-rn}
    For a subspace $E$ of $\mathbb{R}^n$,  parameter $s>0$ and $k \in \mathbb{N}$, we say $f: \mathbb{R}^n \rightarrow [-1,1] \in \Junta_{E,k,s}$ if
    \begin{itemize}
        \item There is a subspace $E' \subseteq E$ of dimension $k$ such that $f \in \Junta_{E'}$.
        \item $f$ is $s$-smooth.
    \end{itemize}
\end{definition}

\begin{definition}~\label{def:max-corr-junta}
    For a function $h: \mathbb{R}^n \rightarrow [-1,1]$, a subspace $E$ of $\mathbb{R}^n$, $k \in \mathbb{N}$ and $ s>0$, we define
    \[
        \rho_{E, k, s}(h) = \max_{\phi \in \Junta_{E,k,s}} \opE_{\bx} [\phi(\bx) \cdot h(\bx)].
    \]
\end{definition}
\begin{definition}~\label{def:rho-c-correlation}
For a function $h: \mathbb{R}^n \rightarrow [-1,1]$ and  a class $\mathcal{C}$ of functions mapping $\mathbb{R}^k \rightarrow [-1,1]$, we define 
\[
\rho_{\mathbb{R}^n, \mathcal{C}} (h) \coloneqq \max_{\phi \in \mathsf{Ind}_{n} (\mathcal{C})}  \mathop{\mathbf{E}}_{\bx} [\phi(\bx) \cdot h(\bx)]. 
\]
For a subspace $E$ of $\mathbb{R}^n$, we define $\mathsf{Ind}_{E,\mathcal{C}}$ to be the set of all functions $\Phi$ which can be expressed in the form 
\[
\Phi(x) = h(\langle v_1, x\rangle) ,\ldots, \langle v_k, x\rangle), 
\]
where $h \in \mathcal{C}$ and $v_1, \ldots, v_k$ are orthonormal vectors in $E$. Thus, $\mathsf{Ind}_{E}(\mathcal{C})$ lifts functions in $\mathcal{C}$ to functions over $\mathbb{R}^n$ where the relevant subspace is $E$. 

For such a class $\mathcal{C}$, a subspace $E$ of $\mathbb{R}^n$ and a function $f: \mathbb{R}^n \rightarrow [-1,1]$, we define 
\[
\rho_{E,\mathcal{C}} (f) \coloneqq \max_{\Phi \in \mathsf{Ind}_{E}(\mathcal{C})} \mathbf{E}_{\bx} [f(\bx) \cdot \Phi(\bx)].
\]
\end{definition}



\subsection{Useful results about  matrices}
\begin{definition}~\label{def:SVD}
    Let $B \in \mathbb{R}^{m \times n}$ matrix. Then, the singular value decomposition (SVD) of $B $ corresponds to $B = U\cdot D \cdot V^T$ where (i) $D \in \mathbb{R}^{r \times r}$ is a diagonal matrix with nonzero entries and (ii) the columns of $U$ and $V$ are orthonormal. The columns of $U$ form an orthonormal basis for the column span of $B$. Similarly, the columns of $V$ form an orthonormal basis for the row span of $B$.
\end{definition}

We will also need  the following random sampling result concerning rank one matrices due to Rudelson and Vershynin~\cite{rudelson2007sampling}.
\begin{theorem}~\label{thm:random-sampling}
    Let $\bZ$ be a distribution over $\mathbb{R}^n$  such that with probability $1$, for $Z \sim \bZ$, we have $\Vert Z \Vert_2 \le M$. Assume that $\Vert \mathbf{E}[ \bZ \otimes \bZ ] \Vert_2 \le 1$. Let $\bZ_1, \ldots, \bZ_d$ be i.i.d.~copies of $\bZ$. Let $a$ be defined as
    \[
        a = C \sqrt{\frac{\log d}{d}} M,
    \]
    for an absolute constant $C>0$. Then,
    \[
        \Pr\bigg[ \bigg\Vert\frac{1}{d} \cdot \big(\sum_{j=1}^d \bZ_j \otimes \bZ_j\big)- \mathbf{E}_{}[\bZ \otimes \bZ] \bigg\Vert_2 >t\bigg] \le 2 e^{-Ct^2/a^2}.
    \]
\end{theorem}
Next, we recall the notion of pseudoinverse of a matrix~\cite{moore_1955, penrose_1955}. Our definition
below is specialized to real square matrices though the definition can be generalized to complex rectangular matrices as well.
\begin{definition}~\label{def:pseudoinverse}
    For any square matrix $A \in \mathbb{R}^{n \times n}$, there is a unique matrix $B$ which satisfies the following conditions (known as the Moore-Penrose conditions):
    \begin{enumerate}
        \item $A B A = A$ and $B A B=B$.
        \item $(AB)^{t} = AB$ and $(BA)^t = BA$.
    \end{enumerate}
    $B$ is referred to as the \emph{pseudoinverse} of $A$. We remark that when $A$ is invertible, then $B = A^{-1}$. We will thus overload this notation and in general, use $A^{-1}$ to denote the pseudoinverse of $A$.
\end{definition}
\begin{claim}~\label{clm:symm-pseudo}
    Let $A \in \mathbb{R}^{m \times m}$ be a symmetric matrix whose non-zero eigenvalues are $\{\lambda_1, \ldots, \lambda_t\}$ and corresponding orthonormal vectors $\{v_1, \ldots, v_t\}$ (note that $t \le m$).
    Then,
    \[
        A = \sum_{i=1}^t \lambda_i v_i v_i^t \ \ \textrm{and} \ \ A^{-1}  = \sum_{i=1}^t \frac{1}{\lambda_i} v_i v_i^t .
    \]
\end{claim}
\begin{proof}
    It is immediate to verify that the Moore-Penrose conditions from Definition~\ref{def:pseudoinverse} hold for $A^{-1}$ defined as above (uses the fact that $v_i$ are orthonormal).
\end{proof}
\begin{definition}~\label{def:matrix-projection}
    For a symmetric matrix $A \in \mathbb{R}^{n \times n}$
    and a parameter $\eta \in \mathbb{R}$, we define
    $A_{\ge \eta} \in \mathbb{R}^{n \times n}$ as projection of $A$ to the eigenspaces with eigenvalue more than $\eta$. In other words, let the spectral decomposition of $A$ be
    $$
        A = \sum_{i=1}^n  \lambda_i v_iv_i^t.
    $$
    Then,
    $$
        A_{\ge \eta} = \sum_{i : \lambda_i \ge \eta} \lambda_i v_i v_i^t.
    $$
    Further, for $\eta>0$, we define $A_{\ge \eta}^{-1}$ is defined to be
    $$
        A_{\ge \eta}^{-1} = \sum_{i : \lambda_i \ge \eta} \frac{1}{\lambda_i} v_i v_i^t.
    $$
    Note that this is the same as the pseudoinverse of $A_{\ge \eta}$. Finally, for a symmetric matrix $A$ and parameter $\eta \in \mathbb{R}$, we let $\Es_{\eta}(A)$ denote $\spann (\{v_i\}_
        {\lambda_i \ge \eta})$.
\end{definition}

\subsection{Algorithmic ingredients}
We will require some algorithmic ingredients
from the paper~\cite{DMN19a}. The first is Lemma~10 in the full version from \cite{DMN19a} which is stated below.
\begin{lemma}~\label{lem:compute-inner-prod-gradient}
    There is an algorithm \textsf{Compute-inner-product} which given oracle access to function $g: \mathbb{R}^n \rightarrow [-1,1]$, noise parameter $t>0$, error parameter $\epsilon>0$, confidence parameter $\delta>0$ and has the following guarantee:
    \begin{enumerate}
        \item It makes $\mathsf{poly}(t, 1/\epsilon,\log(1/\delta))$ queries to $g$.
        \item With confidence $1-\delta$, it outputs $\langle \nabla(P_t g)(y_1) , \nabla(P_t g)(y_2) \rangle$ up to additive error $\pm \epsilon$.
    \end{enumerate}
\end{lemma}
The second lemma we need appears as Lemma~12 in the full version  of \cite{DMN19a} and is stated below.
\begin{lemma}~\label{lem:compute-project-gradient}
    There is an algorithm \textsf{Project-on-gradient} which given oracle access to function $g: \mathbb{R}^n \rightarrow [-1,1]$, noise parameter $t>0$, error parameters $\eta, \nu>0$ and confidence parameter $\delta>0$, has the following guarantee: For any  $x, y \in \mathbb{R}^n$, there is a quantity $\mathsf{Est}(x,y)$ which satisfies
    \[
        \Pr_{\by \sim \gamma_n} [|\mathsf{Est}(x,\by)- \langle \nabla P_t g(x), \by \rangle | > \lambda \eta ] \le \frac{1}{\lambda^2}.
    \]
    The algorithm \textsf{Project-on-gradient} with probability $1-\delta$, outputs an $\pm \nu$-additive estimate of $\mathsf{Est}(x,y)$
    and makes $\mathsf{poly}(1/t,1/\eta,1/\nu, \log(1/\delta))$ queries to $g$.

\end{lemma}

\section{Projection on low-dimensional space and correlation with linear juntas}
The goal of this section is to prove the following theorem.

\begin{theorem}~\label{lem:idealized}
    Let $\Phi: \mathbb{R}^n \rightarrow [-1,1]$ be a (differentiable)
    $C$-Lipschitz function and $\eta, \delta>0$. Let
    $\bx_1,\ldots, \bx_M \sim \gamma_n$
    where $M \ge \frac{C^2}{\eta^2} \log (C \delta/\eta)$. Then, with
    probability $1-\delta$, the matrix $A \in \mathbb{R}^{n \times n}$ defined
    as
    \[
        A = \frac{1}{M} \sum_{j=1}^M \nabla \Phi (\bx_j) \cdot \nabla \Phi (\bx_j)^t,
    \]
    satisfies the following: for every subspace $E$ containing $\mathsf{E}_{\eta/2}(A)$, for every $s\geq 0$,
    and for every $h \in \Junta_{\mathbb{R}^n,k,s}$, we have
    \[
        |\mathbf{E}_{\bx}[\Phi(\bx) \cdot (\calA_E h)(\bx) ] - \mathbf{E}_{\bx}[\Phi(\bx) \cdot h(\bx) ] | \le \sqrt{k \cdot \eta}.
    \]
\end{theorem}

At a high level, this theorem says that for any Lipschitz function $\Phi$, its
correlation with the best linear $k$-junta essentially remains preserved if
we restrict our attention to a subspace obtained by spectrally truncating the
empirical covariance matrix of $\nabla \Phi$.
It is the first step in realizing part I. from Section~\ref{sec:techniques}
(the other step is to handle the fact that we can only estimate $A$).

The proof of Theorem~\ref{lem:idealized} follows from the following lemma.

\begin{lemma}~\label{lem:completeness}
    Let $E$ be a subspace of $\mathbb{R}^n$ and let
    $f: \mathbb{R}^n \rightarrow \mathbb{R}$ be such that for every unit vector
    $v \in E^{\perp}$, $\mathbf{E}[\langle v, \nabla f(\bx) \rangle^2] \le \delta$. 
    Then for every $s\geq 0$, and for every $h \in \Junta_{\mathbb{R}^n,k,s}$, we have
    \begin{equation} ~\label{eq:compare-small-gradient}
        |\mathbf{E}_{\bx}[f(\bx) \cdot (\calA_E h)(\bx)] - \mathbf{E}_{\bx}[f(\bx) \cdot h(\bx)] |\le \sqrt{k \delta}. 
    \end{equation}
\end{lemma}

\begin{proofof}{Theorem~\ref{lem:idealized}}
    Define the matrix $A_{\mathsf{avg}}$ as
    \[
        A_{\mathsf{avg}} = \mathbf{E}_{\bx} [\nabla \Phi(\bx) \cdot \nabla \Phi(\bx)^t].
    \]
    Observe that by Theorem~\ref{thm:random-sampling}, with probability
    $1-\delta$, we have that $\Vert A_{\mathsf{avg}} - A \Vert \le \eta/2$.
    This implies that for any $E \supseteq \mathsf{E}_{\eta/2}(A)$ and unit vector $v \in E^{\perp}$, we have
    $v \in \mathsf{E}_{\eta/2}(A)^\perp$ and hence
    \begin{align}\label{eq:bound-v-corr-A}
        \opE_{\bx} [\langle v, \nabla \Phi(\bx)\rangle^2] = v^T \cdot A_{\mathsf{avg}} \cdot v \le v^T \cdot A \cdot v + \frac{\eta}{2} \le \eta.
    \end{align}
    Then, applying Lemma~\ref{lem:completeness} to the function $h$ and the subspace $E$, we have the proof.
\end{proofof}
We now turn to proving Lemma~\ref{lem:completeness}.
\begin{proofof}{Lemma~\ref{lem:completeness}}
    Let $h\in \Junta_F$ for some subspace $F$ with $\mathsf{dim}(F) \le k$. Let $E' = \spann (E \cup F)$ and define $g = \calA_E h$. Observe that $g$ is $s$-smooth (by Item~2 of Proposition~\ref{it:prop-2}) and thus
    $g \in \Junta_{E,k,s}$. Also, observe that $h = \calA_{E'}h$. We now have
    \begin{align}
        \big|\mathop{\mathbf{E}}_{\bx} [f(\bx) \cdot g(\bx)] - \mathop{\mathbf{E}}_{\bx} [f(\bx) \cdot h(\bx)] \big| & = \big|\mathop{\mathbf{E}}_{\bx} [f(\bx) \cdot \calA_{E} h(\bx)] - \mathop{\mathbf{E}}_{\bx} [f(\bx) \cdot \calA_{E'}h(\bx)] \big| \nonumber                                                                      \\
                                                                                                                     & = \big|\mathop{\mathbf{E}}_{\bx} [\calA_{E}f(\bx) \cdot  h(\bx)] - \mathop{\mathbf{E}}_{\bx} [\calA_{E'}f(\bx) \cdot h(\bx)] \big| \ \textrm{(Item~\ref{it:self-adjoint} of Lemma~\ref{lem:avg-props})} \nonumber \\
                                                                                                                     & \le \big(\mathop{\mathbf{E}}_{\bx} [(\calA_{E}f(\bx) - \calA_{E'}f(\bx))^2]\big)^{\frac12} \ \textrm{(by Cauchy-Schwarz)}. \label{eq:jensen-ae-aep}
    \end{align}
    We now seek to bound the right hand side of \eqref{eq:jensen-ae-aep}. Towards this, let  us split $\mathbb{R}^n = E' \oplus H$ and $E' = E \oplus J$. Here $H$ is the orthogonal complement of $E'$ and $J$ is the orthogonal complement of $E$ inside $E'$.
    For any $x \in \mathbb{R}^n$, we express it as $(x_H, x_J, x_E)$ ($x_J$ represents the component of $x$ along the subspace $J$ and likewise for $H$ and $E$). Observe that for $x=(x_H, x_J, x_E)$, we have
    \begin{align}
        \calA_{E'} f(x) = \mathbf{E}_{\bx'_H} [f(\bx'_H, x_J, x_E)]  \ \textrm{and} \ \calA_{E} f(x) = \mathbf{E}_{\bx'_H, \bx'_J} [f(\bx'_H, \bx'_J, x_E)]. \label{eq:average-formula}
    \end{align}
    Thus, we now have the following:
    \begin{align}
        \mathop{\mathbf{E}}_{\bx} [(\calA_{E}f(\bx) - \calA_{E'}f(\bx))^2] & =  \mathbf{E}_{\bx_H, \bx_J, \bx_E}[ (\mathbf{E}_{\bx'_H} [f(\bx'_H, \bx_J, \bx_E)] - \mathbf{E}_{\bx'_H, \bx'_J} [f(\bx'_H, \bx'_J, \bx_E)])^2
        ] \nonumber                                                                                                                                                                                                              \\
                                                                           & =  \mathbf{E}_{ \bx_J, \bx_E}[ (\mathbf{E}_{\bx'_H} [f(\bx'_H, \bx_J, \bx_E)] - \mathbf{E}_{\bx'_H, \bx'_J} [f(\bx'_H, \bx'_J, \bx_E)])^2 \nonumber \\
                                                                           & \leq \mathbf{E}_{ \bx_H,\bx_J, \bx_E}[ (f(\bx_H, \bx_J, \bx_E) - \mathbf{E}_{ \bx'_J} [f(\bx_H, \bx'_J, \bx_E)])^2]. \label{eq:jensen-2}
    \end{align}
    The last inequality follows from Jensen's inequality.  Next, for any $x= (x_J, x_H, x_E)$, define $f_{x_H, x_E}: \mathbb{R}^J \rightarrow \mathbb{R}$ as $f_{x_H, x_E}(x_J) = f(x_H, x_J,x_E)$. Then,
    \begin{equation}~\label{eq:nabla-projection}
        \Pi_J \nabla f(x) = \nabla f_{x_H, x_E}.
    \end{equation}
    Now applying the definition of $f_{x_H, x_E}$ to \eqref{eq:jensen-2} and subsequently applying the Gaussian Poincar\'{e} inequality, we get
    \begin{align*}
        \mathop{\mathbf{E}}_{\bx} [(\calA_{E}f(\bx) - \calA_{E'}f(\bx))^2] & \le \mathbf{E}_{ \bx_H,\bx_J, \bx_E} [(f_{\bx_H, \bx_E} (\bx_J)  - \mathbf{E}_{\bx_{J}'} [f_{\bx_H, \bx_E} (\bx'_J)])^2]. \\
                                                                           & \le \mathbf{E}_{ \bx_H,\bx_J, \bx_E} [\Vert \nabla f_{\bx_H, \bx_E} (\bx_H, \bx_J, \bx_E)\Vert_2^2].
    \end{align*}
    Finally, applying \eqref{eq:nabla-projection}, we get
    \begin{align*}
        \mathop{\mathbf{E}}_{\bx} [(\calA_{E}f(\bx) - \calA_{E'}f(\bx))^2] & \le \mathbf{E}_{ \bx_H,\bx_J, \bx_E} [\Vert\Pi_J \nabla f (\bx_H, \bx_J, \bx_E)\Vert_2^2].
    \end{align*}
    Now, by our assumption, for any direction $v$ in $J$ (since it is orthogonal to $E$),
    $\mathbf{E}_{\bx} [\Vert \Pi_v \nabla f (\bx) \Vert_2^2 ] \le \delta$. Since the dimension of $J$ is at most $k$, we get that
    \[
        \mathop{\mathbf{E}}_{\bx} [(\calA_{E}f(\bx) - \calA_{E'}f(\bx))^2]  \le k \delta.
    \]
    Combining with \eqref{eq:jensen-ae-aep}, we get the claim.

\end{proofof}

\section{Roadmap for proving Theorem~\ref{thm:main-1}, Theorem~\ref{thm:main-2}, Theorem~\ref{thm:main-4} and Theorem~\ref{thm:C_test}}
In this section, we give a roadmap for our main results -- namely, Theorem~\ref{thm:main-1}  ,  Theorem~\ref{thm:main-2}, Theorem~\ref{thm:main-4} and Theorem~\ref{thm:C_test}. 
First of all, 
observe that instantiating Theorem~\ref{thm:C_test} for the class of linear $k$-juntas with surface area at most $s$ (which are $O(s)$-smooth  by Proposition~\ref{prop:smooth-inclusion}) implies Theorem~\ref{thm:main-1}. 
As mentioned earlier, noise tolerant testing for a class is equivalent to computing the maximum correlation between a function and the same class. Thus, we will prove the  following (equivalent) version of Theorem~\ref{thm:C_test}. 
\begin{theorem}~\label{thm:C_test_new}
For any class $\mathcal{C}$ of functions mapping $\mathbb{R}^k \rightarrow [-1,1]$ (each of which is $s$-smooth), there is an algorithm \textsf{Robust-$\mathcal{C}$-test} which has the following guarantee: 
given error parameter $\epsilon>0$ and oracle access to $f: \mathbb{R}^n \rightarrow [-1,1]$, it outputs an estimate $\hat{\rho}_{\mathbb{R}^n, \mathcal{C}}(f)$ such that
\[
\big|\hat{\rho}_{\mathbb{R}^n, \mathcal{C}} (f)-{\rho}_{\mathbb{R}^n, \mathcal{C}}(f)\big| \le \epsilon. 
\]
The query complexity is $k^{\mathsf{poly}(s/\epsilon)}$. 
\end{theorem} 
Note that by instantiating Theorem~\ref{thm:C_test_new} with the class of $s$-smooth functions, we get Theorem~\ref{thm:main-2}.  Finally, we note that the proof of Theorem~\ref{thm:C_test_new} can be easily modified to yield Theorem~\ref{thm:main-4}. This is explained in Section~\ref{sec:linvar}. Thus, we now focus on proving Theorem~\ref{thm:C_test_new} (which is equivalent to Theorem~\ref{thm:C_test}).

To do this, our first step is to replace the function $f$
by a smoothed version:
\begin{lemma}~\label{lem:smooth-Lip}
    For smoothness parameter $s$, error parameter $\kappa>0$ and $f: \mathbb{R}^n \rightarrow [-1,1]$, the function $\fsm$
    defined by $\fsm = P_{\kappa^2/s^2} f$ has the following guarantees:
    \begin{enumerate}
        \item $f \in \mathcal{C}^{\infty}$ and $f$ is $L$-Lipschitz for  $L = O(s^2/\kappa^2)$.
        \item For any $x \in \mathbb{R}^n$, $\fsm(x)$ can be computed to error $\eta/10$ with probability $1-\delta$ using $T(\eta,\delta) = \poly(1/\eta,\log(1/\delta))$ queries to the oracle for $f: \mathbb{R}^n \rightarrow [-1,1]$.
        \item
              Let $g: \mathbb{R}^n \rightarrow [-1,1]$ be a $s$-smooth function. Then, 
              \[
                  \big| \mathbf{E}_{\bx} [\fsm (\bx) g(\bx)] - \mathbf{E}_{\bx} [f (\bx) g(\bx)]\big| \le \frac{\kappa}{2}.
              \]
    \end{enumerate}
\end{lemma}
\begin{proof}
    The first property follows from Fact~\ref{fact:P_t-lipschitz} and the definition of the
    noise operator $P_t$.
    The second property follows easily from the definition of $P_t$:
    we simply have to take enough samples to estimate the expectation.
    Finally, suppose $g$ is a $s$-smooth function. Then, it follows that $\mathbf{E} [| P_{\kappa^2/s^2} g (\bx) - g(\bx) |] = O(\kappa)$. It follows that 
    \begin{align*}
    \big| \mathbf{E}_{\bx} [\fsm (\bx) g(\bx)] - \mathbf{E}_{\bx} [f (\bx) g(\bx)] \big| &=  \big| \mathbf{E}_{\bx} [ P_{\kappa^2/s^2} f (\bx) g(\bx)] - \mathbf{E}_{\bx} [f (\bx) g(\bx)] \big| \\ &= \big| \mathbf{E}_{\bx} [( P_{\kappa^2/s^2} g (\bx) - g(\bx)) \cdot f (\bx) ] \big| \\ 
    &\le O(\kappa).  
    \end{align*} 
    
%
%
\end{proof}
Using Lemma~\ref{lem:smooth-Lip}, it suffices to prove Theorem~\ref{thm:main-2} for Lipschitz functions. In particular, we shall prove the following version of Theorem~\ref{thm:main-2} for Lipschitz functions.
\begin{theorem}~\label{thm:main-2-Lipschitz}
    For any class $\mathcal{C}$, there is an algorithm \textsf{Correlation-smooth-junta-$\mathcal{C}$}  with the following guarantee:
    Let $\fsm:\mathbb{R}^n \rightarrow [-1,1]$ be an infinitely differentiable $L$-Lipschitz function such that $\fsm = P_u f$ for a parameter $u>0$ (where $f:\R^n \rightarrow [-1,1]$). The algorithm
    is given oracle access to the functions $\fsm$ and $f$. It also gets as inputs,  error parameter $\epsilon>0$, junta arity parameter $k$ and outputs an
    estimate $\hat{\rho}_{\mathbb{R}^n,\mathcal{C}}(\fsm)$ (with probability at least $2/3$) with the following guarantee:
    \[
        |\hat{\rho}_{\mathbb{R}^n,\mathcal{C}}(\fsm) - {\rho}_{\mathbb{R}^n,\mathcal{C}}(\fsm)| \leq \epsilon.
    \]
   Here ${\rho}_{\mathbb{R}^n,\mathcal{C}}(\fsm)$ is the maximum correlation of $\fsm$ with any $s$-smooth $k$-linear junta. The query complexity of the algorithm is $\mathsf{poly}(L/u) \cdot k^{O(s^2/\epsilon^2)}$.
    Further, the algorithm also works even when we have a noisy oracle to $\fsm$ -- in particular, the above guarantee holds even when each evaluation of $\fsm(\cdot)$ at $x$ returns $\pm \eta$ additive error estimate for $\eta=\mathsf{poly}(u/L) \cdot k^{O(-s^2/\epsilon^2)}$.
\end{theorem}
To obtain Theorem~\ref{thm:C_test_new}, we let $\kappa = \epsilon/4$, $u = \kappa/s$. Define $\fsm = P_u f$. We now invoke Theorem~\ref{thm:main-2-Lipschitz} on $\fsm$ witth error parameter $\epsilon/2$ -- observe that the output $\hat{\rho}_{\mathbb{R}^n,\mathcal{C}}(\fsm)$ satisfies 
\[
|\hat{\rho}_{\mathbb{R}^n,\mathcal{C}}(\fsm) - {\rho}_{\mathbb{R}^n,\mathcal{C}}(\fsm)| <\epsilon. 
\]
Finally, observe that while we do not have oracle access to $\fsm$, Theorem~\ref{thm:main-2-Lipschitz} only requires to evaluate $\fsm(\cdot)$ with an additive error of $\pm \eta = \mathsf{poly}(u/L) \cdot k^{O(-s^2/\epsilon^2)}$. Observe that the number of queries made by Theorem~\ref{thm:main-2-Lipschitz} is $\mathsf{Q} = \mathsf{poly}(L/u) \cdot k^{\Theta(s^2/\epsilon^2)}$. Set $\delta = 1/(10 \mathsf{Q})$. Using Lemma~\ref{lem:smooth-Lip},  we can evaluate $\fsm(x)$ by making $\eta^{-2}\log(1/\delta)$  to the oracle for $f$. For our choice of $\delta$, this means that with probability $9/10$, all our evaluations of $\fsm(\cdot)$ are $\pm \eta$ accurate. This means that we can simulate our queries to $\fsm$ by using the oracle for $f$ with a multiplicative ovehead of $\eta^{-2} \log (1/\delta)$. Plugging in the values of $\eta$ and $\delta$, we get the final claim. 


\section{Proof of Theorem~\ref{thm:main-2-Lipschitz}}
We now turn to the proof of Theorem~\ref{thm:main-2-Lipschitz}. For the moment, we will just assume that we can evaluate $\fsm$ at any point $x$ exactly. From the description of our algorithm, it would be clear that the guarantee of algorithm continues to hold even if each evaluation of $\fsm(x)$ has an additive error of $ \pm \eta = \mathsf{poly}(u) \cdot k^{O(-L^2/\epsilon^2)}$. We will bring this to attention of the reader at the relevant points.
The algorithm \textsf{Correlation-smooth-junta} invokes two crucial subroutines. The first is the routine \textsf{Implicit projection} described in  Figure~\ref{fig:DS}.

\subsection{Implicit projection algorithm}

\begin{figure}[tb]
    \hrule
    \vline
    \begin{minipage}[t]{0.98\linewidth}
        \vspace{10 pt}
        \begin{center}
            \begin{minipage}[h]{0.95\linewidth}
                {\small
                    \underline{\textsf{Inputs}}}
                \vspace{5 pt}

                \begin{tabular}{ccl}
                    $f$        & := & Oracle access to function $f: \mathbb{R}^{n} \rightarrow [-1,1]$ \\
                    $\fsm$     & := & Oracle access to function $\fsm:
                        \mathbb{R}^n \rightarrow [-1,1]$ where $\fsm = P_u f$.                         \\
                    $L$        & := & Lipschitz parameter                                              \\
                    $\nu$ & := & accuracy parameter                                               \\
                    $k$        & := & junta arity parameter                                            \\
                \end{tabular}

                {\small
                \underline{\textsf{Parameters}}}
                \vspace{5 pt}

                \begin{tabular}{ccl}
                    $\delta$    & := & $\frac{1}{20}$                                                                     \\
                    $M$         & := & $\frac{L^2}{\eta^2} \log (L\delta/\eta)$                                           \\
                    $\eta$      & := & $\frac{\nu^2}{100k}$.                                                         \\
                    $\epsilon'$ & := & $\frac{\eta^5  \nu^2}{L^8  C_0^2 M^6}$ (where $C_0$ is a large absolute constant - $10^6$ suffices for us) \\
                \end{tabular}

                \vspace{5 pt}
                \underline{\textsf{Implicit projection algorithm}}
                \begin{enumerate}
                    \item  Sample $M$ random points $\bx_1, \ldots, \bx_M \sim \gamma_n$.
                    \item For each $1 \le i, j \le M$, with confidence parameter $\delta/M^2$ and error parameter $\epsilon'$, we compute $\langle \nabla \fsm (\bx_i), \nabla \fsm (\bx_j) \rangle = \langle \nabla P_u f (\bx_i), \nabla P_u f (\bx_j) \rangle$ using algorithm
                          \textsf{Compute-inner-product} from  Lemma~\ref{lem:compute-inner-prod-gradient}.  Denote this by $\widehat{A}_{i,j}$ and let $\widehat{A} \in \mathbb{R}^{M \times M}$ as the corresponding symmetric matrix.
                    \item Let $\widehat{N}$ be the closest psd matrix to $\widehat{A}$ in Frobenius norm (can be computed using convex programming).
                    \item Let $\widehat{V} \widehat{D}^2 \widehat{V}^T$ be the spectral decomposition of $\widehat{N}$.
                    \item Output the points $(\bx_1, \ldots, \bx_M)$ and the matrix $\widehat{W}= \widehat{D}^{-1}_{\ge \sqrt{\eta}/2} \cdot \widehat{V}^T$.
                \end{enumerate}

                \vspace{5 pt}

            \end{minipage}
        \end{center}

    \end{minipage}
    \hfill \vline
    \hrule
    \caption{Description of the  testing algorithm \textsf{Implicit projection}}
    \label{fig:DS}
\end{figure}

\begin{lemma}~\label{lem:implicit-projection}
The algorithm \textsf{Implicit projection} takes as input oracle access to $f:\mathbb{R}^n \rightarrow [-1,1]$ and $\fsm :\mathbb{R}^n \rightarrow [-1,1]$, parameters $u, L>0$, error parameter $\nu>0$ and junta arity parameter $k$. Suppose $\fsm = P_u f$. The algorithm makes $\poly(k,1/u,1/\nu,L)$ queries to $f$ and $\fsm$ and with probability $9/10$,  has the following guarantee: For $M = \mathsf{poly}(k/\nu)$, it outputs $M$ points $x_1, \ldots, x_M$ and a matrix $\hat{W} \in \mathbb{R}^{M \times M}$. Let $B^T  \in \mathbb{R}^{M \times n}$ be the matrix whose $j^{th}$ row is $\nabla \fsm (x_j)$ and $\hat{E}$ be the span of the rows of $\hat{W}B^T$. There exists a $k$-dimensional subspace $\tilde{E}$ of $\hat{E}$ with the following property. 
 Let $h \in \Junta_{\mathbb{R}^n,k,s}$.  Then, for $g=\calA_{\tilde{E}}h$, 
\begin{equation}~\label{eq:fsm-correlation}
|\mathbf{E}_{\bx}[\fsm(\bx) g(\bx)] - \mathbf{E}_{\bx}[\fsm(\bx) h(\bx)]| \le \frac{\nu}{2}. 
\end{equation}
Further, the matrix $\hat{W}$ satisfies 
\begin{equation}~\label{eq:diffe-e-pi}
\Vert \Pi_{\hat{E}} - B \hat{W}^T \hat{W} B^T\Vert_F  = \Vert \hat{I} -  \hat{W} B^T \hat{W}^T\Vert_F \le \nu/2,
\end{equation}
where $\hat{I}$ denotes the identity matrix in $M$ dimensions. Finally, the matrix $\hat{W}$ satisfies $\Vert \hat{W} \Vert_2 \le \frac{20k}{\nu}$. 
\end{lemma}
The high level idea of the lemma is the following: Let $E$ denote the subspace spanned by the rows of $B^T$.  Let us define 
 $N = B^TB$ and $\widehat{\Pi} =  B \widehat{W}^T \widehat{W} B^T$. To understand the high level idea behind the algorithm \textsf{Implicit projection}, observe that if in Step~2, we could compute $\langle \nabla \fsm(\bx_i), \nabla \fsm(\bx_j) \rangle$ exactly, then $\hat{N} =N$. Conseuqently, if $\eta>0$ is sufficiently small, then it is easy to see that the rows of $\hat{W} B^T$ form an orthonormal basis of $E$ and consequently, $\hat{\Pi}$ is a projection matrix into $E$. Unfortunately for us, we will not have access to $B$ explicitly and thus are only able to compute an approximation to $N$, namely $\hat{N}$. The goal here is two-fold: (a) Understand why the rows of $\hat{W}B^T$  are essentially orthonormal; (b)  show that for $g=\calA_{\hat{E}}h$, $\mathbf{E}_{\bx}[\fsm(\bx) g(\bx)]$ is nearly as large as $\mathbf{E}_{\bx}[\fsm(\bx) h(\bx)]$.



The next  claim quantifies the sense in which the rows of $\hat{W}B^T$ are almost orthonormal.  
\begin{lemma}\label{lem:almost-isometry}
    For matrices $\hat D$, $\hat W$, $B$, $\hat N$ and $\eta >0$ (as described in the algorithm \textsf{Implicit projection}), let $\hat I = \hat D_{\ge \sqrt{\eta}/2}^{-1} \hat D$. (That is, $\hat I$ has a 1 
    corresponding to large eigenvalues of $\hat N$.) Let $\hat E$ be the span of the
    rows of $\hat W B^T$. Then
    \[
        \|\Pi_{\hat E} - B \hat W^T \hat W B^T\|_F = \|\hat I - {\hat W B^T} \cdot   B \hat W^T\|_F \le \frac 4\eta \|\hat N - B^T B\|_F.
    \]
\end{lemma}

\begin{proof}
    Since $\hat N = \hat V \hat D^2 \hat V^T$, we can write
    \[
        \hat I = (\hat D_{\ge \sqrt \eta/2})^{-1} \hat V^T \hat N \hat V (\hat D_{\ge \sqrt \eta/2})^{-1}.
    \]
    Then
    \[
        \hat I - \hat W B^T B \hat W^T = 
        (\hat D_{\ge \sqrt \eta/2})^{-1} \hat V^T (\hat N - B^T B) \hat V (\hat D_{\ge \sqrt \eta/2})^{-1}.
    \]
    Finally, note that $\|(\hat D_{\ge \sqrt \eta/2})^{-1}\| \le \frac{2}{\sqrt \eta}$
    and $\|AB\|_F \le \|A\|_F \|B\|$ for any matrices $A$ and $B$.
    This proves the claimed inequality. To see the equality, 
    note that $B \hat W^T \hat W B^T$ and $\hat W B^T B \hat W^T$ have the same
    eigenvalues, and both expressions can be expressed as $(\sum (\lambda_i - 1)^2)^{1/2}$,
    where the sum ranges over non-zero eigenvalues.
\end{proof}
Having shown that the rows of $\hat{W}B^T$ are close to being orthonormal, we next show that the rows of 
$\hat{W} B^T$ essentially span $\Es_{\ge \eta}(BB^T)$ -- more precisely, we show that $\Pi_{\Es_{\ge \eta}(B B^T)} (I - B \hat W^T \hat W B^T)$  is small. 

\begin{lemma}\label{lem:almost-same-subspace}
For matrices $\hat D$, $\hat W$, $B$, $\hat N$ and $\eta >0$ (as described in the algorithm \textsf{Implicit projection}),
    \[
        \|\Pi_{\Es_{\ge \eta}(B B^T)} (I - B \hat W^T \hat W B^T)\| \le \frac{20 \Vert B^TB \Vert_F \cdot \|B^T B\|}{\eta^{\frac52}} \sqrt{\|\hat N - B^T B\|}.
    \]
\end{lemma}

\begin{proof}
    Recall that $B = U D V^T$ is a singular value decomposition of $B$. Let $U_{\ge \sqrt\eta}$ consist of the rows of $U$ whose singular
    values are at least $\sqrt \eta$, so that
    \begin{equation}~\label{eq:def-pi-es}
        \Pi_{\Es_{>\eta}(BB^T)} = U_{\ge \sqrt \eta} U_{\ge \sqrt \eta}^T = U I_{\ge \sqrt \eta} U^T.
    \end{equation}
    \begin{align*}
        \|\Pi_{E_{\ge \eta}(B B^T)} (I - B \hat W^T \hat W B^T)\|
        &= \|\Pi_{E_{\ge \eta}(B B^T)} (I - B (\hat N_{\ge \eta/4})^{-1} B^T)\| \\
        &= \|U I_{\ge \sqrt \eta} U^T - U I_{\ge \sqrt \eta} D V^T (\hat N_{\ge \eta/4})^{-1} V D U^T)\| \\
        &= \|V I_{\ge \sqrt \eta} V^T - V I_{\ge \sqrt \eta} D V^T (\hat N_{\ge \eta/4})^{-1} V D V^T)\| \\
        &= \|V I_{\ge \sqrt \eta} V^T - (N_{\ge \eta})^{1/2} (\hat N_{\ge \eta/4})^{-1} N^{1/2} \|. 
    \end{align*}
    where in the last line we set $N = B^T B$.  The first equality uses 
    $\hat W^T \hat W = (\hat N_{\ge \eta/4})^{-1}$. The second and third equality uses that $\Vert A \Vert  = \Vert \Lambda A \Lambda^T \|$ for unitary matrix $\Lambda$ and the last equality sets $N=V D V^T$. Now, observe that  $V I_{\ge \sqrt \eta} V^T = (N_{\ge \eta})^{1/2} (N_{\ge \eta})^{-1} N^{1/2}$,
    we have
    \begin{align}
        \|\Pi_{\Es_{\ge \eta}(B B^T)} (I - B \hat W^T \hat W B^T)\|
        &= \|(N_{\ge \eta})^{1/2} \Big( (N_{\ge \eta})^{-1} - (\hat N_{\ge \eta/4})^{-1}\Big) N^{1/2} \| \nonumber \\
        &\le \|N\| \cdot  \|\Pi_{\Es_{\ge \eta}(N)} \Big((N_{\ge \eta})^{-1} - (\hat N_{\ge \eta/4})^{-1}\Big) \| ~\label{eq:N-subspace-bound}
    \end{align}
    Finally, we apply Lemma~\ref{lem:pseudoinverse-stab} to get that 
    \[
    \|\Pi_{\Es_{\ge \eta}(N)} \Big((N_{\ge \eta})^{-1} - (\hat N_{\ge \eta/4})^{-1}\Big) \| \le \frac{20 \Vert N \Vert_F \sqrt{\Vert N - \hat N \Vert}}{\eta^{5/2}}. 
    \]
    Combining this with \eqref{eq:N-subspace-bound}, we
    get the result.  
\end{proof}

\begin{lemma}\label{lem:good-subspace-finally}
Let $\fsm : \mathbb{R}^n \to [-1, 1]$, $L$, $M$, $\eta$, $k$, $\delta$, $\hat{N}$, $\hat{W}$ and $B$ be as described in the  Algorithm \textsf{Implicit projection}. Let $\hat{E}$ denote the span of the rows of $\hat{W}B^T$. 
If $\fsm$ is $L$-Lipschitz, with probability $1-\delta$,  there is a subspace $\tilde{E}$ of $\hat{E}$ with the following property: For all $h: \mathbb{R}^n \rightarrow [-1,1]$ such that $h \in \Junta_{\mathbb{R}^n,k,s}$, 
    \[
        |\mathbf{E}_{\bx}[\fsm(\bx) \cdot \calA_{\tilde{E}} h(\bx)] - \mathbf{E}_{\bx}[\fsm(\bx) \cdot  h(\bx)]| \le \sqrt{k \cdot \eta} - \frac{80L^4 \cdot M^{5/2}}{\eta^{5/2}} \sqrt{\|\hat N - B^T B\|} - \frac{16 L}{\eta} \|\hat N - B^T B\|_F.
    \]
\end{lemma}

\begin{proof}
Define $E = E_{\ge \eta/2} (BB^T)$. Let $\hat{E}$ be the span of the rows of $\hat{W}B^T$.
Then
    \begin{align}
        \|\Pi_E \Pi_{\hat E^\perp}\|
        &= \| \Pi_E - \Pi_E \hat \Pi + \Pi_E (\hat \Pi - \Pi_{\hat E})\| \nonumber \\
        &\le \|\Pi_E - \Pi_E \hat \Pi\| + \|\hat \Pi - \Pi_{\hat E}\| \nonumber \\
        &\le \frac{20}{\eta^{5/2}} \Vert B^T B \Vert_F \Vert B^TB \Vert\sqrt{\|\hat N - B^T B\| } + \frac{4}{\eta} \|\hat N - B^T B\|_F, \label{eq:bound-proj-diff-E-Eperp}
    \end{align}
    where the last line follows from Lemmas~\ref{lem:almost-same-subspace} and~\ref{lem:almost-isometry}. 
    We now apply Corollary~\ref{cor:correlation-subspace-stability} to the above to get that there is a subspace $\tilde{E}$ of $\hat{E}$ such that for all $h: \mathbb{R}^n \rightarrow [-1,1]$, 
    \begin{align}
        |\mathbf{E}_{\bx} [\fsm(\bx) \cdot \calA_{\tilde{E}} h(\bx)] - \mathbf{E}_{\bx} [\fsm(\bx) \cdot \calA_{E} h(\bx)]| &\le 4L  \Vert \Pi_E \Pi_{(\hat{E})^{\perp}} \Vert_F \nonumber \\
    &\le \frac{80 L \cdot\sqrt{M}}{\eta^{5/2}} \Vert B^T B \Vert_F\Vert B^TB \Vert \sqrt{\|\hat N - B^T B\| } - \frac{16 L}{\eta} \|\hat N - B^T B\|_F \label{eq:comp-project-Ehat-E}. 
    \end{align}
    On the other hand, since $h \in \Junta_{\mathbb{R}^n,k,s}$,  using Theorem~\ref{lem:idealized}, with probability $1-\delta$, we also have 
    \[
    |\mathbf{E}_{\bx} [\fsm(\bx) \cdot \calA_{E} h(\bx)] -\mathbf{E}_{\bx} [\fsm(\bx) \cdot  h(\bx)]| \le \sqrt{k \cdot \eta}. 
    \]
    Combining the above with \eqref{eq:comp-project-Ehat-E}, we get that with probability $1-\delta$, there is a subspace $\tilde{E}$ of $\hat{E}$ such that for all $h \in \Junta_{\mathbb{R}^n,k,s}$ (mapping $\mathbb{R}^n$ to $[-1,1]$), 
    \begin{align}
    |\mathbf{E}_{\bx} [\fsm(\bx) \cdot \calA_{\tilde{E}} h(\bx)] - \mathbf{E}_{\bx} [\fsm(\bx) \cdot  h(\bx)]| \le \sqrt{k \cdot \eta} -  \frac{80 L \cdot\sqrt{M}}{\eta^{5/2}} \Vert B^T B \Vert_F\Vert B^TB \Vert \sqrt{\|\hat N - B^T B\| } - \frac{16 L}{\eta} \|\hat N - B^T B\|_F. 
    \end{align}
    Now, observe that $\Vert B^T B \Vert_F \le ML^2$ (Since $\fsm$ is $L$-Lipschitz, each row has norm at most $L$). This then implies the claim.



\end{proof}

\begin{proofof}{Lemma~\ref{lem:implicit-projection}}
By our setting of parameters, observe that with probability $1-\delta$, the matrix $\hat{A}$ satisfies $\Vert \hat{A} - B^TB \Vert_\infty \le \epsilon'$. This in turn implies that 
$\Vert \hat{A} - B^TB \Vert_F \le \epsilon' \cdot M$. Since $\hat{N}$ is the closest psd matrix to $\hat{A}$, this means $\Vert N - \hat{N} \Vert_F \le 2 \epsilon' \cdot M$.

 Plugging the values of $\epsilon'$, 
$\eta$ and $M$ into Lemma~\ref{lem:good-subspace-finally} shows that 
\eqref{eq:fsm-correlation} is satisfied with probability at least  $1-2\delta = 9/10$. 
 Similarly, \eqref{eq:diffe-e-pi} follows by plugging the values of $\epsilon'$, $\eta$ and $M$ into Claim~\ref{lem:almost-isometry}. Finally, observe that the query complexity of the algorithm  
is dictated by Step~2 (i.e., the query complexity of the routine \textsf{Compute-inner-product}). By plugging in Lemma~\ref{lem:compute-inner-prod-gradient}, we get that the query complexity is $\mathsf{poly}(M,1/u,1/\epsilon')$. Plugging in the values of these parameters (from the description of the algorithm \textsf{Implicit projection}), we get the claim. 

Finally, to get an upper bound on $\Vert \hat{W} \Vert_2$, observe that $\hat{W} = \hat{D}_{\ge \sqrt{\eta}/2} \cdot \hat{V}$. This means that $\Vert \hat{W} \Vert_2 \le 2/\sqrt{\eta}$. Plugging in the value of $\eta$ from the description of \textsf{Implicit projection}, we get the claim. 
\end{proofof}

\subsection{The averaged class}

We next describe a preprocessing step for our class of functions $\calC$.
The point is that it is possible in principle for $f$ to be an $E$-Junta
but be well-correlated with some function $g \in \calC$ that was embedded in
$\R^n$ along a \emph{different} subspace $\tilde E$. We handle this by adding to
$\calC$ all possible projections of functions from $\calC$ that were embedded
in different subspaces.

\begin{definition}\label{def:class-averaging}
    For a class $\calC$ of functions $\R^k \to [-1, 1]$, define $\calC^\ast$
    to be the set of all functions $\R^k \to [-1, 1]$ of the form
    \[
        x \mapsto \E_{\bz \sim \gamma_k} [g(W^T (\begin{smallmatrix}x \\ \bz\end{smallmatrix})))],
    \]
    where $g$ ranges over $\calC$ and $W$ ranges over all $(2k) \times k$ matrices
    with orthonormal columns.
\end{definition}
In other words, we are taking functions from $\calC$, embedding them in $\R^{2k}$
along an arbitrary $k$-dimensional subspace, and then averaging them back down to
$\R^k$. As a consequence of Proposition~\ref{prop:smooth-inclusion}, if every 
function is $\calC$ is $s$-smooth, then so is every function in $\calC^\ast$.
Also, $\calC^\ast$ contains $\calC$, as can be seen by taking the first $k$
rows of $A$ to be an orthonormal basis of $\R^k$, and the next $k$ rows to be zero. We next have the  following claim. 
\begin{claim}~\label{clm:class-averaging} 
Let $\calC$ be a class of functions mapping $\R^k \to [-1, 1]$. Define the set $\mathcal{F}$ to be the functions of the form $\calA_{\mathbb{R}^m} f$ where $f\in \mathsf{Ind}_{n}(\calC)$ (where $n \ge m+k$ and $m \ge k$). 
Then, $\mathcal{F} = \mathsf{Ind}_{m}(C^\ast)$. 

\end{claim}
\begin{proof}
It is easy to see that $\mathsf{Ind}_{m}(C^\ast) \subseteq \mathcal{F}$ as long as $n \ge m+k$. 
So, we now argue that $\mathcal{F} \subseteq \mathsf{Ind}_{m}(C^\ast)$. Let $f \in \mathsf{Ind}_n(\mathcal{C})$. 
Let $E$ be the relevant subspace for $f$ and let $E = J \oplus J'$ where $J = \mathbb{R}^m \cap E$ and $J'$ is the orthogonal complement of $J$ inside $E$. It is obvious that the dimension of $J$ is at most $k$. It now easily follows that $g \in \mathsf{Ind}_m(C^\ast)$.

\end{proof}

We remark that although it might be challenging in general to characterize $\calC^\ast$
in terms of $\calC$, there are several classes of functions where this is easy:
\begin{itemize}
    \item if $\calC$ is the class of all $s$-smooth functions then $\calC^\ast = \calC$;
    \item if $\calC$ is the class of all half-spaces then $\calC^\ast$ is the class of all functions
        of the form $x \mapsto \Phi(\inr{a}{x} + b)$, where $\Phi$ is the Gaussian c.d.f.;
    \item more generally, if $\calC$ is closed under taking subspaces -- in the sense that
        if $g \in \calC$, $E \subset \R^k$ is a subspace, and $z \in E^\perp$ then $x \mapsto f(\pi_E x + z)$
        also belongs to $\calC$ -- then $\calC^\ast$ is contained in the convex hull of $\calC$.
        In this situation, and because we will be interested in maximizing a linear function over $\calC$,
        we can essentially replace $\calC^\ast$ by $\calC$ in what follows.
\end{itemize}

%

\subsection{Hypothesis testing on low-dimensional space}

Our final technical task is to show that functions on a low-dimensional space can be adequately ``pulled back''
to $\R^n$ under an approximate projection. The first observation is that an approximate projection
can be approximated by a projection:

\begin{lemma}\label{lem:approximate-projection}
    For any $m \le n$ and any $m \times n$ matrix $X$ of rank $m$, there exists an $m \times n$ matrix $Y$ with orthogonal
    rows, such that
    \[
        \|X - Y\|_F \le \|X X^T - I\|_F.
    \]
\end{lemma}

\begin{proof}
    Let $U D^2 U^T = X X^T$ be a singular value decomposition of $X X^T$. Then
    $I = (D^{-1} U^T X) (D^{-1} U^T X)^T$, and it follows that $V^T := D^{-1} U^T X$
    is an orthogonal matrix. Let $Y = U V^T$. Noting that $X = U D V^T$, we have
    $\|X - Y\|_F^2 = \|D - I\|_F^2$, and if $\sigma_1, \dots, \sigma_m$ are the singular
    values of $X$ then
    \[
        \|D - I\|_F^2 = \sum (\sigma_i - 1)^2 \le \sum (\sigma_i^2 - 1)^2 = \|X X^T - I\|_F^2.
    \]
\end{proof}

\subsubsection{The existence of a small net}

We now prove the existence of a small net of Lipschitz functions for families of $s$-smooth
$k$-linear Juntas in $\R^m$. The main result is Proposition~\ref{prop:net-existence-1}.

We begin with a few preliminaries related to approximate by Lipschitz functions, namely, $s$-smooth
functions can be approximated by Lipschitz functions and Lipschitz functions don't change much
under composition by nearby linear maps.

\begin{lemma}\label{lem:approx-by-lipschitz}
    For every $s$-smooth function $f: \R^n \to [-1, 1]$ and every $\epsilon > 0$, there is a $\frac{C \cdot s}{\epsilon}$-Lipschitz
    function $g: \R^n \to [-1, 1]$ such that $\|f - g\|_{L^2(\gamma)} \le \epsilon$. Here $C$ is the absolute constant appearing in Fact~\ref{fact:P_t-lipschitz}. 
\end{lemma}

\begin{proof}
    Choose $t = \frac{\epsilon^2}{s^2}$ and set $g = P_t f$, so that the bound $\|f - g\|_{L^2(\gamma)} \le \epsilon$
    follows from the fact that $f$ is $s$-smooth. The claim follows from Fact~\ref{fact:P_t-lipschitz}. 
\end{proof}

\begin{lemma}\label{lem:composition-distance}
    Suppose that $g: \R^m \to \R$ is Lipschitz and let $X$ and $Y$ be two $m \times n$ matrices. Then
    $\|g \circ X - g \circ Y\|_{L^2(\gamma)} \le (\Lip g) \|X - Y\|_F$ (here $\Lip g$ denotes the Lipschitz constant of $g$).
\end{lemma}

\begin{proof}
    Let $\bx$ be a standard normal random variable on $\R^n$. Then
    \[
        \E[ ((g \circ X)(\bx) - (g \circ Y)(\bx))^2] \le (\Lip g)^2 \E[ \|X \bx - Y \bx\|^2] = (\Lip g)^2 \|X - Y\|_F^2.
    \]
\end{proof}

Our procedure for producing a net for $s$-smooth $k$-juntas in $\R^m$ proceeds in three steps. First,
we will construct a net for $s$-smooth functions on $\R^k$. Then we will find a net for
$k$-dimensional subspaces of $\R^m$. Combining these two nets will give a net for
$s$-smooth $k$-juntas in $\R^m$.



We begin with the net for $s$-smooth functions on $\R^k$. Before we do, we recall the following simple fact. 
\begin{fact}~\label{fact:net}
    For the unit sphere in $\mathbb{R}^{m}$ (denoted by $\mathbb{S}^{m-1}$), there is a $\delta$-net (in Euclidean sphere) of size $(1/\delta)^{O(m)}$.
\end{fact}
\begin{lemma}\label{lem:cover-for-smooth}
    For any $k \in \N$ and any $s, \epsilon > 0$, there exists a set $\Net$ of functions
    $\R^k \to [-1, 1]$ such that
    \begin{itemize}
        \item[(1)] every function in $\Net$ is $\frac{Cs}{\epsilon}$-Lipschitz (here $C$ is the absolute constant appearing in Fact~\ref{fact:P_t-lipschitz}),
        \item[(2)] $\Net$ is an $\epsilon$-net for the set of $s$-smooth functions $\R^k \to [-1, 1]$,
        \item[(3)] $\log |\Net| \le k^{O(s^2/\epsilon^2)}$, and
        \item[(4)] Every function $f$ in $\Net$ is $s$-smooth.  
    \end{itemize}
\end{lemma}

\begin{proof}
    We first construct a set $\Net$ which satisfies properties (2), (3) and (4) (in fact, the functions in $\Net$ will be $s$-smooth). Once we achieve this, for every $g \in \Net$, we will include $P_u g$ in $\Net$ and discard $g$. Observe that since $g$ is $s$-smooth, for $u = \epsilon^2/s^2$, $P_u g$  $1/\sqrt{u} = O(s/\epsilon)$-Lipschitz. Further, observe that the property of being $s$-smooth is closed under the noise operator $P_u$. 
    Thus, properties (1)-(4) are then simultaneously satisfied.

    We now turn to construction of set $\Net$ satisfying (2) and (3) such that every function in $\Net$ is $s$-smooth. To do this,
    we assume that the reader is familiar with the basics of Hermite analysis (see Chapter~11 of \cite{ODonnell:book}).  In particular, recall that the Hermite polynomials over $\mathbb{R}^k$ are indexed by $S \in (\mathbb{N}^+)^k$ (where $\mathbb{N}^+ = \mathbb{N} \cup \{0\}$). Further, every $f \in L_2(\gamma_k)$ can be represented as
    \[
        f(x) = \sum_{S \in (\mathbb{N}^+)^k} \widehat{f}(S) H_S(x) \ \  \textrm{and} \ \ P_tf(x) = \sum_{S \in (\mathbb{N}^+)^k} \widehat{f}(S) e^{-t\Vert S \Vert_1}H_S(x).
    \]
    Let $\delta>0$ which we will fix later. Set $t = \delta^2/s^2$.
    \begin{align*}
        \sum_{S} \widehat{f}^2(S) (1-e^{-t\Vert S \Vert_1}) = \mathbf{E}_{\bx} [ f(\bx) \cdot ( f(\bx) - P_t f(\bx) )] & \le \mathbf{E}_{\bx}[\Vert f(\bx) - P_t f(\bx) \Vert] \leq \delta.
    \end{align*}
    Here the  equality follows from Parseval's identity, the first inequality uses the fact that the range of $f$ is $[-1,1]$ and the second inequality uses the assumption that $f$ is $s$-smooth.

    Set $m = 1/t$ and this means that
    $
        \sum_{S:\Vert S \Vert_1 \ge m} \widehat{f}^2(S) \cdot \big(1-\frac{1}{e} \big) \le \delta.
    $
    Consequently, we have
    \[
        \sum_{S:\Vert S \Vert_1 \ge m} \widehat{f}^2(S) \le \delta \cdot \frac{e}{e-1} < 4\delta.
    \]
    Let us define $f_{\mathsf{tr}} = \sum_{S:\Vert S \Vert_1 \le m} \widehat{f}(S) H_S(x)$. This means that for $f_{\mathsf{tr}} : \mathbb{R}^k \rightarrow \mathbb{R}$, $\Vert f_{\mathsf{tr}} - f \Vert_{L_2(\gamma)}^2 < 4 \delta$. Next, we recall that the unit ball in $\ell_2^M$ admits a $\delta$-net of size $(1/\delta)^{O(M)}$ (Fact~\ref{fact:net}). Since the cardinality of the set $\{S :\Vert S \Vert_1 \le m\}$ is at most $k^m$, we get that there is a set of functions $\Net_{\mathsf{tr}}$ with the following properties:
    \begin{enumerate}
        \item Every $g \in \Net_{\mathsf{tr}}$ satisfies $g:\mathbb{R}^k \rightarrow \mathbb{R}$ and $| \Net_{\mathsf{tr}}| \le (1/\delta)^{O(k^m)}$.
        \item There exists some $g \in \Net_{\mathsf{tr}}$ such that $\mathbf{E}_{\bx}[(g(\bx) - f_{\mathsf{tr}}(\bx))^2] \le \delta$. This implies that $\mathbf{E}_{\bx}[(g(\bx) - f_{}(\bx))^2] \le 10\delta$.
    \end{enumerate}
    Finally, we set $\delta = \epsilon/40$ which shows that for every $s$-smooth $f: \mathbb{R}^k \rightarrow [-1,1]$, there is a function $g \in \Net_{\mathsf{tr}}$ such that $\Vert g - f \Vert_{L_2(\gamma)}^2 \le \epsilon/4$. Further, $\log |\mathsf{Net}_{\mathsf{tr}}| \le k^{O(s^2/\epsilon^2)} \log (1/\epsilon)$.

    Observe that the functions in the set $\mathsf{Net}_{\mathsf{tr}}$ are
    not necessarily bounded. To obtain bounded $s$-smooth functions (i.e.,
    the set $\mathsf{Net}$), for each $g \in \mathsf{Net}_{\mathsf{tr}}$, we
    choose an arbitrary $s$-smooth $h: \mathbb{R}^k \rightarrow [-1,1]$ such
    that $\Vert g - h \Vert_{L_2(\gamma)}^2 \le \epsilon/4$ and include it in
    $\mathsf{Net}$. Note that the size of $\mathsf{Net}$ is no larger than
    $\mathsf{Net}_{\mathsf{tr}}$. Further, from the property of $\mathsf{Net}_{\mathsf{tr}}$, we have that for every $s$-smooth $f$, there is a function $g  \in \mathsf{Net}$ such that
    $ \Vert g - f \Vert_{L_2(\gamma)} \le \epsilon/2$. This finishes the proof.
\end{proof}


Next, we need to turn our net of functions on $\R^k$ into a net of $k$-linear-juntas on $\R^m$.
We will do this by finding an appropriate net for $k$-dimensional subspaces of $\R^m$, and then
using the net of Lemma~\ref{lem:cover-for-smooth} for each of these subspaces.

\begin{lemma}\label{lem:subspace-net}
    There is a set $\calE$ of $k$-dimensional subspaces of $\R^m$ such that
    \begin{enumerate}
        \item for every $k$-dimensional subspace $E$ of $\R^m$, there is some $E' \in \calE$
              with $\|\Pi_E - \Pi_{E'}\|_F \le \epsilon$; and
        \item $|\calE| \le \left(\frac{O(k)}{\epsilon}\right)^{mk}$. 
    \end{enumerate}
\end{lemma}

\begin{proof}
    Let $T$ be a $\delta$-net of $\mathbb{S}^{m-1}$ (the unit Euclidean sphere in $\R^m$)
    of cardinality at most $(1/\delta)^{O(m)}$ (as described in Fact~\ref{fact:net}).
    Let $\calE$ be the set of all $k$-dimensional subspaces that are spanned
    by $k$ elements of $T$. The claimed bound on the cardinality of $\calE$ follows,
    provided we choose $\delta$ so that $\epsilon \le C' k \delta$ (for an absolute constant $C'$).

    Let $E$ be a $k$-dimensional subspace of $\R^m$, and let $x_1, \dots, x_k$ be an orthonormal basis
    of $E$. Choose $y_1, \dots, y_k \in T$ with $\|x_i - y_i\| \le \delta$ for all $i$; then
    the $y_i$ are unit vectors, and for $i \ne j$ we have
    \[
        |\inr{y_i}{y_j}| = |\inr{y_i}{y_j} - \inr{x_i}{x_j}| \le |\inr{x_i}{x_j - y_j}| + |\inr{y_j}{x_i - y_i}|
        \le 2 \delta.
    \]
    It follows that if $Y$ is the matrix with rows $y_i$, and if $E'$ is the span of $y_1, \dots, y_k$,
    then $\|Y^T Y - \Pi_{E'}\|_F^2 = \|Y Y^T - I\|_F^2 \le 4 \delta^2 k$.
    Hence,
    \[
        \|\Pi_E - \Pi_{E'}\|_F = \|X^T X - \Pi_{E'}\|_F \le \|X^T X - Y^T Y\|_F + 2 \delta \sqrt k.
    \]
    It remains to bound $\|X^T X - Y^T Y\|_F$, and it will suffice to show that $\|X^T X - Y^T Y\|_F =O(k \delta)$.

    Now, if $x$ and $y$ are unit vectors with $\|x - y\| \le \delta$, then $\inr{x}{y} \ge 1 - O(\delta^2)$.
    It follows that $\|x x^T - y y^T\|_F^2 = 2 - 2 \inr{x}{y}^2 \le O(\delta^2)$. Thus, by the triangle inequality,
    \[
        \|X^T X - Y^T Y\|_F \le \sum_{i=1}^k \|x_i x_i^T - y_i y_i^T\|_F = O( k \delta).
    \]
\end{proof}


\begin{definition}
Let $\hat{E}$ be a $m$-dimensional subspace of $\mathbb{R}^n$ and let $\mathcal{C}$ be a subset of 
$s$-smooth linear $k$-juntas over $\mathbb{R}^k$. 
We define $\mathsf{Ind}_{\hat{E}}(\mathcal{C})$ to be the set of all functions $h: \mathbb{R}^n \rightarrow \mathbb{R}$ of the form 
\[
\Phi(x) = h(\langle v_1, x \rangle, \ldots, \langle v_k,x \rangle), 
\]
where $v_1, \ldots, v_k$ are orthonormal vectors in $\hat{E}$. In other words, $\mathsf{Ind}_{\hat{E}}(\mathcal{C})$ lifts the functions in $\mathcal{C}$  to linear $k$-juntas over $\mathbb{R}^n$ where the relevant subspace is a $k$-dimensional subspace of $\hat{E}$. Note that $\mathsf{Ind}_{\mathbb{R}^n}(\mathcal{C}) = \mathsf{Ind}_{n}(\mathcal{C})$ (see Definition~\ref{def:induced-class}).

 Finally, for such a class $\mathcal{C}$, subspace $\hat{E}$ and a function $f: \mathbb{R}^n \rightarrow [-1,1]$, 
\[
\rho_{\hat{E},\mathcal{C}} (f) \coloneqq \max_{\Phi \in \mathsf{Ind}_{\hat{E}}(\mathcal{C})} \mathbf{E}_{\bx} [\Phi(\bx) \cdot f(\bx)]. 
\]
\end{definition} 

\begin{proposition}\label{prop:net-existence-1}
    Let $\mathcal{C}$ be a subset of $s$-smooth linear functions $\mathbb{R}^k \to [-1, 1]$. Then, for any $m \ge k$, there is a set $\mathsf{Net}_{m,\mathcal{C}}$ of functions mapping $\mathbb{R}^m$ to $[-1,1]$ which  satisfies the following properties:
\begin{enumerate}
\item Any $g \in \mathsf{Net}_{m,\mathcal{C}}$ is $Cs/\epsilon$-Lipschitz.
\item $\mathsf{Net}_{m,\mathcal{C}}$ is an $\epsilon$-net for $\mathsf{Ind}_{\mathbb{R}^m}(\mathcal{C})$ -- i.e., for every $h \in \mathsf{Ind}_{\mathbb{R}^m}(\mathcal{C})$, there is a $g \in \mathsf{Net}_{m,\mathcal{C}}$ such that $\Vert h-g\Vert_{L^2(\gamma)} \le \epsilon$. 
\item $\log |\Net_{m,\mathcal{C}}| \le k^{O(s^2/\epsilon^2)} + O\big( mk \log \frac{k s}{\epsilon}\big)$, and
        \item For every function in $g \in \Net_{m,\mathcal{C}}$, there is $h \in \mathsf{Ind}_{\mathbb{R}^m}(\mathcal{C})$ such that $\Vert h-g\Vert_{L^2(\gamma)} \le \epsilon$.

\end{enumerate} 
\end{proposition} 


\begin{proof}
    It suffices to consider the case that $\calC$ is the set of \emph{all} $s$-smooth functions
    $\R^k \to \R$. Indeed, once we have found a net (call it $\mathsf{Net}_0$) for this case,
    we can handle the case of general $\calC$ by simply discarding any 
    $g \in \mathsf{Net}_0$ for which there is no $h \in \mathsf{Ind}_{\mathbb{R}^m}(\mathcal{C})$ such that 
$\Vert h-g\Vert_{L^2(\gamma)} \le \epsilon$. In this way, we ensure that property 4 is satisfied,
noting that properties 1, 2 and 3 remain unchanged if we remove functions $g$ from $\mathsf{Net}_0$. For the
rest of this proof, we consider the case that $\calC$ is the set of all $s$-smooth functions.

    Let $\widehat \Net$ be a net for $s$-smooth functions on $\R^k$, with the properties guaranteed by
    Lemma~\ref{lem:cover-for-smooth}. Let $\calE$ be a collection of $k$-dimensional subspaces
    of $\R^m$, with the properties guaranteed by Lemma~\ref{lem:subspace-net} with accuracy $\epsilon' = \epsilon^2/s$.
    We define $\Net_0$ to be the set of functions of the form $x \mapsto f(\Pi_E x)$, where $f \in \widehat  \Net$ and $E \in \calE$.
    Clearly, $\Net_0$ satisfies Property~1. To see Property 3, note that 
    $\log |\Net_0| = \log |\widehat \Net| + \log |\calE|$.  By using Lemma~\ref{lem:cover-for-smooth} and  
    Lemma~\ref{lem:subspace-net}, the bound on $\log |\Net_0|$ follows.  Thus, it remains to show Property~2.

    To see Property 2, suppose that $f$ is an $s$-smooth $k$-Junta. Then there is some $k$-dimensional
    subspace $E$ and an $s$-smooth function $g$ on $\R^k$ such that $f = g \circ \Pi_E$. Choose $h \in \widehat \Net$
    to be $\epsilon$-close to $g$ and choose $E' \in \calE$ such that $\|\Pi_E - \Pi_{E'}\|_F \le \epsilon^2/s$.
    Then $h \circ \Pi_{E'}$ belongs to $\Net$, and satisfies
    \[
        \|h \circ \Pi_{E'} - f\|_{L^2(\gamma)}
        \le \|h \circ \Pi_{E'} - h \circ \Pi_E\|_{L^2{\gamma}} + \|h \circ \Pi_{E} - f \circ \Pi_E\|_{L^2(\gamma)}
    \]
    The second term is at most $\epsilon$, and the first term can be bounded (using Lemma~\ref{lem:composition-distance})
    by $(\Lip h) \|\Pi_{E'} - \Pi_E\| \le \frac{Cs}{\epsilon} \cdot \epsilon^2/2 \le C \epsilon$.
    This proves the claim (after we change $\epsilon$ by a constant factor).
\end{proof}

\subsubsection{Proof of the theorem}

Finally, here is the application of Proposition~\ref{prop:net-existence-1} to the analysis of our
algorithm.

\begin{lemma}~\label{lem:approximaiton-lip}
    Let $\mathcal{C}$ be a subset of $s$-smooth functions $\mathbb{R}^k \to [-1, 1]$ and 
let $\Net_{\mathcal{C},m}$ be an $\epsilon$-net as guaranteed
    by Proposition~\ref{prop:net-existence-1}.
Let $\hat{E}$ be a $m$-dimensional subspace of $\mathbb{R}^n$ and let $ A \in \mathbb{R}^{m \times n}$ with the following two properties: 
(i) the rows of $A$ span $\hat{E}$ and (ii) $\Vert AA^T - I \Vert_F \le \kappa$.    Then, for any Lipschitz function $\fsm$, we have that 
\[
\big| \rho_{\hat{E},\mathcal{C}} (\fsm) - \max_{h \in \Net_{m,\mathcal{C}}} \E_{\bx} [h(A \bx) \cdot \fsm(\bx)] \big| \le \frac{C' s \kappa}{\epsilon}  + \epsilon, 
\]
for an absolute constant $C'$. 
\end{lemma}
\begin{proof}
 Let $Y$ be an $m \times n$ matrix with orthonormal rows whose rows also span $\hat{E}$ such that $\Vert Y -A \Vert_F \le \Vert AA^T - I \Vert_F$. For any $h \in \mathsf{Net}_{m,\mathcal{C}}$, observe that $h$ is $O(s/\epsilon)$-Lipschitz. Thus,  we have,     \begin{align*}
        \big|\E_{\bx} [h( A \bx) \cdot \fsm(\bx)] - \E_{\bx} [h (Y \bx) \cdot \fsm(\bx)]\big|
         & \le \|h \circ A  - h \circ Y\|_{L^2(\gamma)]} \ \textrm{(using Cauchy Schwartz inequality)}\\
         & \le \frac{O(1) \cdot s}{\epsilon} \|A - Y\|_F  \ \textrm{(using Lemma~\ref{lem:composition-distance})}               \\
          & \le \frac{O(1) \cdot s}{\epsilon} \|AA^T -I\|_F  \ \textrm{(using Lemma~\ref{lem:approximate-projection})}               \\
         & \le \frac{O(1) s\kappa}{\epsilon}. 
    \end{align*}
Thus, to prove the claim, it suffices to show 
that
    \begin{equation}~\label{eq:ref-final-eq}
        \big| \rho_{\hat E, \mathcal{C}}(\fsm) - \max_{h \in \Net_{m,\mathcal{C}}} \E_{\bx} [h(Y \bx) \cdot \fsm(\bx)] \big| \le \epsilon.
    \end{equation}
Now, recall that 
$$
\rho_{\hat E, \mathcal{C}}(\fsm) \coloneqq  \max_{h \in \mathsf{Ind}_{\hat{E}}(\mathcal{C})} \mathbf{E}_{\bx} [h(\bx) \cdot \fsm(\bx)] = \max_{h \in \mathsf{Ind}_{\mathbb{R}^m}(\mathcal{C})} \mathbf{E}_{\bx} [h(Y \bx) \cdot \fsm(\bx)]. 
$$    
    The second equality is an easy consequence of the definition of the class $\mathsf{Ind}_{\hat{E}}(\mathcal{C})$. Thus, we get that 
    \begin{align*}
     \rho_{\hat E, \mathcal{C}}(\fsm) - \max_{h \in \Net_{m,\mathcal{C}}} \E_{\bx} [h(Y \bx) \cdot \fsm(\bx)]    &=  \max_{h \in \mathsf{Ind}_{\mathbb{R}^m}(\mathcal{C})} \mathbf{E}_{\bx} [h(Y \bx) \cdot \fsm(\bx)] - \max_{h \in \Net_{m,\mathcal{C}}} \E_{\bx} [h(Y \bx) \cdot \fsm(\bx)] .   
    \end{align*}
To upper bound the right hand side, let $\arg \max_{h \in \mathsf{Ind}_{\mathbb{R}^m}(\mathcal{C})} \mathbf{E}_{\bx} [h(Y \bx) \cdot \fsm(\bx)] = h_\ast$. Then, note that there is a function $h_\ast^{'} \in \Net_{m,\mathcal{C}}$ such that $\Vert h_\ast - h_\ast^{'} \Vert_{L^2(\gamma)} \le \epsilon$ (Property 1 in Proposition~\ref{prop:net-existence-1}). By Cauchy-Schwartz, we get that the right hand side is upper bounded by $\epsilon$. Similarly, we can also lower bound the right hand side by $-\epsilon$, exploiting Property~4 in Proposition~\ref{prop:net-existence-1}. This implies \eqref{eq:ref-final-eq}.

\end{proof}

\begin{proofof}{Theorem~\ref{thm:main-2-Lipschitz}}

Let $\calC^\ast$ be the averaged class of $\calC$, as in Definition~\ref{def:class-averaging}.
Let $\mathsf{Net}_{m,\mathcal{C}^\ast}$ be the set of functions guaranteed by Proposition~\ref{prop:net-existence-1} -- with smoothness parameter $s$, error parameter $\epsilon/4$ and $m=M$ as instantiated in the algorithm \textsf{Implicit Projection}.  Let us also set $\nu = \epsilon^2/(100 C's)$ for the constant $C'$ appearing in Lemma~\ref{lem:approximaiton-lip}. 
Let us now invoke algorithm \textsf{Implicit projection} with smoothness parameter $u = \nu/s$, 
Lipschitz parameter $L=O(s/\nu)$, the error parameter $\nu$  and junta arity parameter $k$.

Suppose $h_\ast \in \mathsf{Ind}_{n}(\mathcal{C})$ such that 
\[
h_\ast = \arg \max_{h \in \mathsf{Ind}_{n}(\mathcal{C})}\mathbf{E}_{\bx} [\fsm(\bx) h(\bx)]. 
\]
Lemma~\ref{lem:implicit-projection} guarantees that with probability $9/10$, we get a matrix $\hat{W}$ and points 
$\bx_1, \ldots, \bx_M$ such that the following conditions are satisfied: 
 let $B^T$ be the matrix where the $j^{th}$ row is $\nabla \fsm(\bx_j)$. Let $\hat{E}$ be the row  span of $B^T$. 
 \begin{enumerate}
 \item $\| \hat{I} - \hat{W}B^T B\hat{W}^T \|_F \le \nu/2$. Here $\hat{I}$ is the identity matrix in $m$ dimensions where $m = \mathsf{dim}(\hat{E})$. 
 \item For $g = \calA_{\hat E}h_\ast$,
 \begin{equation}~\label{eq:closeness-emp-1}
     |\mathbf{E}[\fsm(\bx) \cdot g(\bx)] - \mathbf{E}[\fsm(\bx) \cdot h_\ast(\bx)]| \le \frac{\nu}{2}.
 \end{equation}
 \end{enumerate}

 Also, by Lemma~\ref{lem:approximaiton-lip}, we have that 
 \begin{equation}~\label{eq:closeness-emp-2}
 \big| \rho_{\hat{E},\mathcal{C}^\ast}(\fsm) - \mathop{\max}_{h \in \mathsf{Net}_{m,\mathcal{C}^\ast}} \mathop{\mathbf{E}}_{\bx}  [h(\hat{W}B^T \bx) \cdot \fsm (\bx)]  \big| \le \frac{C's}{\epsilon} \cdot \frac{\nu}{2} + \frac{\epsilon}{4} 
< \frac{51 \cdot \epsilon}{200}
 \end{equation}
 Next, since $g = \calA_{\hat{E}}h_\ast \in \Ind_{\hat E}(\calC^*)$ (by Claim~\ref{clm:class-averaging}), we have that 
\begin{equation}~\label{eq:errbound-1}
\rho_{\hat{E},\mathcal{C}^\ast}(\fsm) \ge \mathbf{E}[\fsm(\bx) \cdot g(\bx)]
\ge \mathbf{E}[\fsm(\bx) \cdot h_\ast(\bx)] - \frac{\nu}{2},
\end{equation} 
where the second inequality follows from~\eqref{eq:closeness-emp-1}.
On the other hand, if $\tilde g \in \Ind_{\hat E}(\calC^\ast)$ maximizes
the correlation with $\fsm$, there exists (again by
Claim~\ref{clm:class-averaging}) $\tilde h_\ast \in \Ind_n(\calC)$
with $\calA_{\hat E} \tilde h_\ast = \tilde g$,
and hence (by Lemma~\ref{lem:implicit-projection})
\begin{equation}~\label{eq:errbound-2}
\rho_{\hat{E},\mathcal{C}^\ast}(\fsm)
= \mathbf{E}[\fsm(\bx) \cdot (\calA_{\hat E}\tilde h_\ast)(\bx)]
\le \mathbf{E}[\fsm(\bx) \cdot \tilde h_\ast(\bx)] + \frac{\nu}{2}
\le \mathbf{E}[\fsm(\bx) \cdot h_\ast(\bx)] + \frac{\nu}{2}. 
\end{equation} 
Together with~\eqref{eq:errbound-1}, we have
\[
 | \rho_{\hat{E},\mathcal{C}^\ast}(\fsm)
    - \mathbf{E}[\fsm(\bx) \cdot h_\ast(\bx)]| \le \frac{\nu}{2};
\]
combined with \eqref{eq:closeness-emp-2} (and recalling that we chose $h_\ast$ to be
a correlation-maximizer in $\Ind_n(\calC)$, we have
\begin{equation}\label{eq:approximation-by-max-over-net}
    \big|\rho_{\R^n,\calC}- \mathop{\max}_{h \in \mathsf{Net}_{m,\mathcal{C}^\ast}} \mathop{\mathbf{E}}_{\bx}  [h(\hat{W}B^T \bx) \cdot \fsm (\bx)]  \big| \le \frac{51\epsilon}{100} + \frac{\nu}{2} =\frac{52 \epsilon}{100}.
\end{equation}

Thus, for our purposes, it suffices to (approximately) compute $\mathop{\max}_{h \in \mathsf{Net}_{m, C^\ast}} \mathop{\mathbf{E}}_{\bx}  [h(\hat{W}B^T \bx) \cdot \fsm (\bx)]$. 
Towards this, consider any fixed $h \in \mathsf{Net}$. We set $T = O(\epsilon^{-2} \log (1/\zeta))$   where $\zeta = 1/(10\cdot |\mathsf{Net}_{m, C^\ast}|)$. Sample $T$ points from the standard Gaussian $\gamma_n$ -- call these points $\bz_1, \ldots, \bz_T$. By applying the Chernoff bounds, observe that for any $h \in \mathsf{Net}$, with probability $1-\zeta$, 
\[
\bigg| \mathbf{E}_{\bx} [h(\hat{W}B^T \bx) \cdot \fsm (\bx)  - \frac{1}{T} \sum_{j=1}^T h(\hat{W}B^T \bz_j) \cdot \fsm(\bz_j) \bigg| \le \epsilon/4. 
\]
From a union bound, it follows that with probability $9/10$, 
\begin{equation}~\label{final-eq-2}
\bigg| \max_{h \in \mathsf{Net}_{m, C^\ast}} \mathbf{E}_{\bx} [h(\hat{W}B^T \bx) \cdot \fsm (\bx)]  -
\max_{h \in \mathsf{Net}_{m, C^\ast}} \frac{1}{T} \sum_{j=1}^T h(\hat{W}B^T \bz_j) \cdot \fsm(\bz_j) \bigg| \le \epsilon/4. 
\end{equation}
Combining \eqref{final-eq-2} and \eqref{eq:approximation-by-max-over-net}, we get
\begin{equation}\label{eq:approximation-by-empirical-max-over-net}
\bigg| 
\rho_{\mathbb{R}^n,\calC}(\fsm)
-
\max_{h \in \mathsf{Net}_{m, C^\ast}} \frac{1}{T} \sum_{j=1}^T h(\hat{W}B^T \bz_j) \cdot \fsm(\bz_j) \bigg| < \frac{2\epsilon}
{3}. 
\end{equation}
Thus, it suffices to compute the quantity 
\[
\mathsf{Corr} = \max_{h \in \mathsf{Net}_{m, C^\ast}} \frac{1}{T} \sum_{j=1}^T h(\hat{W}B^T \bz_j) \cdot \fsm(\bz_j), 
\]
up to additive error $\pm \epsilon/3$ and upper bound the query complexity of computing this estimate. Observe that computing $\{\fsm(\bz_j)\}_{j=1}^T$ requires $T$ queries. Using Lemma~\ref{lem:implicit-projection}, we have 
\[
\Vert \hat{W} \Vert_2 \le \frac{20k}{ \nu} \coloneqq \Delta 
\]
 Set $\theta =\frac{\epsilon^2}{200 \cdot C \cdot\Delta \cdot s \cdot \sqrt{m}}$. Here $C$ is the constant 
 appearing in Fact~\ref{fact:net}.
 We now invoke algorithm \textsf{Project-on-gradient} from Lemma~\ref{lem:compute-project-gradient}. Then, we get that for any $\bz_j$ (for $1\le j \le T$), 
 \[
 \Pr_{\bx_i \sim \gamma_n} [|\mathsf{Est}(\bx_i, \bz_j) - \langle \nabla \fsm(\bx_i), \bz_j \rangle| > \theta] \le \frac{1}{200 T \cdot m}. 
 \]
 Further, we can compute $\pm \theta$ estimate to $\mathsf{Est}(\bx_i, z_j)$ (with confidence $1-\frac{1}{200 T \cdot m}$) where the query complexity is $\mathsf{poly}(T \cdot m, 1/\theta)$. This means that with probability $0.99$, for each $1 \le j \leq T$ and $1 \le i \le m$, we have $\pm 2\theta$ estimates (denoted by $\bchi_{i,j}$) for each $\langle \nabla \fsm(\bx_i), \bz_j \rangle$. In other words, for each $1 \le j \leq T$, we get a vector $\boldsymbol{\Xi}_{j}$ which satisfies 
 \[
 \Vert \boldsymbol{\Xi}_{j} - B^T \mathbf{z}_j \Vert \le 2\theta \sqrt{m}.  
 \]
 Since $\Vert \hat{W} \Vert \le \Delta$, this means that for all $1\le j \le T$, 
 \[
 \Vert  \hat{W}\boldsymbol{\Xi}_{j} - \hat{W} B^T \mathbf{z}_j \Vert \le 2\theta \sqrt{m} \Delta = \frac{\epsilon^2}{100 C\cdot s}.
 \]
 Since $h \in \mathsf{Net}$ is $Cs/\epsilon$-Lipschitz, this implies that for each $1 \le j \le T$, 
 \[
 \big| h(\hat{W}\boldsymbol{\Xi}_{j}
 ) - h(\hat{W} B^T \mathbf{z}_j) \big|  \le  \frac{\epsilon}{100}.
 \]
Consequently, this gives a $\pm \epsilon/100$ additive estimate of the quantity 
\[
\frac{1}{T} \sum_{j=1}^T h(\hat{W}B^T \bz_j) \cdot \fsm (\bz_j). 
\]
Recalling~\eqref{eq:approximation-by-empirical-max-over-net},
 we have shown that the algorithm produces a $\pm \epsilon$-additive estimate of $\rho_{\mathbb{R}^n,\calC}(\fsm)$. It remains to bound the query complexity of the algorithm. The query complexity of the algorithm \textsf{Implicit projection} (from Lemma~\ref{lem:implicit-projection}) is $\mathsf{poly}(k,1/u, 1/\nu, L)$ where $\nu =\epsilon^2/(100 C's)$ ($C'$ is the constant appearing in Lemma~\ref{lem:approximaiton-lip}). Thus, the query complexity of this part is $\mathsf{poly}(k,s,L,1/\epsilon)$.
 
 For the hypothesis testing part,  the query complexity can be bounded as follows: 
 \begin{enumerate}
 \item We make $T$ queries to $\fsm$ where $T =O(\epsilon^{-2} \log |\mathsf{Net}|$. 
 \item For each $1 \le j \le m$ and $1 \le i \le T$, we compute a $\pm \theta$ approximation to $\mathsf{Est}(\bx_i,  \mathbf{z}_j)$ -- the query complexity of each is $\mathsf{poly}(T \cdot m, 1/\theta)$. 
 \end{enumerate}
 Thus, the total query complexity is bounded by $\mathsf{poly}(T, m, 1/\theta)$. Using the fact that $m \le M$ (where $M$ is set in algorithm \textsf{Implicit projection}) and plugging in the value of the parameters, we get the final bound on the query complexity.

 Finally, we remark that our analysis so far was based on assuming that we have exact oracle access to $\fsm$. However, we only have oracle access to $f$ and approximate oracle to $\fsm$ (via Lemma~\ref{lem:smooth-Lip}). To address this issue, we observe that the 
  algorithm \textsf{Implicit projection} only uses the oracle to $f$ and not to $\fsm$ (the only invocation of these oracles is when we call the routine \textsf{Compute-inner-product}). In the hypothesis testing part, (i) we only use the oracle to $f$ when we invoke the algorithm \textsf{Project-on-gradient}. (ii) 
  we use the oracle for $\fsm$ when we approximate $\mathsf{Corr}$ to error $\pm \epsilon/3$. However, it is easy to see that for this, it suffices to have an oracle for $\fsm$ with (say) $O(\epsilon^{-1.5})$ additive accuracy. By Lemma~\ref{lem:smooth-Lip}, this can be simulated with an oracle for $f$ with $O(\epsilon^{-3})$ overhead given an oracle to $f$. This finishes our proof.

\end{proofof}

\section{Learning the linear-invariant structure}~\label{sec:linvar}

The proof of Theorem~\ref{thm:main-4} is essentially the same as the proof of
Theorem~\ref{thm:main-2-Lipschitz}; we construct the same net of functions
and estimate the correlations of each of them. The only difference is that
instead of outputting the maximum correlation
value of a function in the net, we output the set of functions that have
a large correlation.

\begin{proofof}{Theorem~\ref{thm:main-4}}
    Let $\Net_{m,\calC^\ast}$ be as in the proof of Theorem~\ref{thm:main-2-Lipschitz}.
    With the same $\delta$
    as in that proof, with probability $9/10$ we can simultaneously estimate
    $\E_x[h(\hat W B^T \bx) \cdot \fsm(\bx)]$ to error $\pm \epsilon/8$ for all $h \in \Net_{m,\calC^\ast}$.

    Now consider the algorithm that
    returns all $h \in \Net_{m,\calC^\ast}$ for which our estimate of 
    $\E_x[h(\hat W B^T \bx) \cdot \fsm(\bx)]$ is at least $\rho - 4\epsilon$; call the
    returned set $\calG$.
    It follows that for every $\hat g \in \calG$,
    \[
     \E_x[\hat g(\hat W B^T \bx) \cdot \fsm(\bx)] \ge \rho - 5\epsilon,
    \]
    and so the first claim of
    the theorem follows.

    For the second claim, take any $g \in \Ind_n(\calC)$ and let $\hat E$ be the range of
    $\hat W B^T$. Since (by Claim~\ref{clm:class-averaging}) $\calA_{\hat E} g \in \Ind_{\hat E}(\calC^\ast)$,
    there is some $\hat g \in \Net_{m,\calC^\ast}$ such that
    \begin{equation}\label{eq:close-to-the-averaged-func}
        \mathbf{E}_\bx [(\hat g(\hat W B^T \bx) - \calA_{\hat E} g)^2] \le \epsilon^2.
    \end{equation}
    Now, if $g$'s correlation with $f$ is at least $\rho - \epsilon$,
    then by Lemma~\ref{lem:implicit-projection} $\calA_{\hat E} g$ has correlation
    at least $\rho - 2\epsilon$ with $f$, and so by~\eqref{eq:close-to-the-averaged-func},
    $\hat g \circ (\hat W B^T)$
    has correlation with $f$ at least $\rho - 3\epsilon$,
    and hence the definition of $\calG$ ensures that $\hat g \in \calG$.
    Going back to~\eqref{eq:close-to-the-averaged-func},
    the function $\hat g$ witnesses the second claim of the theorem.
\end{proofof}

\bibliographystyle{plain}
\bibliography{allrefs}

\begin{thebibliography}{10}

\bibitem{ba2014deep}
Jimmy Ba and Rich Caruana.
\newblock Do deep nets really need to be deep?
\newblock In {\em {Advances in Neural Information Processing Systems}}, pages
  2654--2662, 2014.

\bibitem{balcan2012active}
Maria-Florina Balcan, Eric Blais, Avrim Blum, and Liu Yang.
\newblock Active property testing.
\newblock In {\em 2012 IEEE 53rd Annual Symposium on Foundations of Computer
  Science}, pages 21--30. {IEEE}, 2012.

\bibitem{belgolsud98}
Mihir Bellare, Oded Goldreich, and Madhu Sudan.
\newblock Free bits, {PCP}s, and nonapproximability--towards tight results.
\newblock {\em SIAM Journal on Computing}, 27(3):804--915, 1998.

\bibitem{bhatia2013matrix}
Rajendra Bhatia.
\newblock {\em Matrix analysis}, volume 169.
\newblock {Springer Science \& Business Media}, 2013.

\bibitem{blais2008improved}
E.~Blais.
\newblock Improved bounds for testing juntas.
\newblock In {\em Approximation, Randomization and Combinatorial Optimization.
  Algorithms and Techniques}, pages 317--330. Springer, 2008.

\bibitem{blais2009testing}
E.~Blais.
\newblock Testing juntas nearly optimally.
\newblock In {\em Proceedings of the forty-first annual ACM symposium on Theory
  of computing}, pages 151--158. ACM, 2009.

\bibitem{blais2018tolerant}
E.~Blais, C.~Canonne, T.~Eden, A.~Levi, and D.~Ron.
\newblock Tolerant junta testing and the connection to submodular optimization
  and function isomorphism.
\newblock In {\em {Proceedings of the Twenty-Ninth Annual ACM-SIAM Symposium on
  Discrete Algorithms}}, pages 2113--2132. Society for Industrial and Applied
  Mathematics, 2018.

\bibitem{Blum:94}
A.~Blum.
\newblock Relevant examples and relevant features: Thoughts from computational
  learning theory.
\newblock in AAAI Fall Symposium on `Relevance', 1994.

\bibitem{BFK+:97}
A.~Blum, A.~Frieze, R.~Kannan, and S.~Vempala.
\newblock A polynomial time algorithm for learning noisy linear threshold
  functions.
\newblock {\em Algorithmica}, 22(1/2):35--52, 1997.

\bibitem{BlumLangley:97}
A.~Blum and P.~Langley.
\newblock Selection of relevant features and examples in machine learning.
\newblock {\em Artificial Intelligence}, 97(1-2):245--271, 1997.

\bibitem{bucilua2006model}
Cristian Bucilu\"{a}, Rich Caruana, and Alexandru Niculescu-Mizil.
\newblock Model compression.
\newblock In {\em Proceedings of the 12th ACM SIGKDD international conference
  on Knowledge discovery and data mining}, pages 535--541, 2006.

\bibitem{chakraborty2012junto}
S.~Chakraborty, E.~Fischer, D.~Garc{\'\i}a-Soriano, and A.~Matsliah.
\newblock Junto-symmetric functions, hypergraph isomorphism and crunching.
\newblock In {\em {27th Annual Conference on Computational Complexity (CCC)}},
  pages 148--158. IEEE, 2012.

\bibitem{CLSSX18}
X.~Chen, Z.~Liu, Rocco~A. Servedio, Y.~Sheng, and J.~Xie.
\newblock Distribution free junta testing.
\newblock In {\em Proceedings of the ACM STOC 2018}, 2018.

\bibitem{chen2017sample}
Xi~Chen, Adam Freilich, Rocco~A Servedio, and Timothy Sun.
\newblock {Sample-Based High-Dimensional Convexity Testing}.
\newblock In {\em {Approximation, Randomization, and Combinatorial
  Optimization. Algorithms and Techniques (APPROX/RANDOM 2017)}}, 2017.

\bibitem{Chen:2017:SQC}
Xi~Chen, Rocco~A. Servedio, Li-Yang Tan, Erik Waingarten, and Jinyu Xie.
\newblock Settling the query complexity of non-adaptive junta testing.
\newblock In {\em Proceedings of the 32Nd Computational Complexity Conference},
  pages 26:1--26:19, 2017.

\bibitem{daniely2016complexity}
Amit Daniely.
\newblock Complexity theoretic limitations on learning halfspaces.
\newblock In {\em {Proceedings of the forty-eighth annual ACM symposium on
  Theory of Computing}}, pages 105--117, 2016.

\bibitem{DMN19a}
Anindya De, Elchanan Mossel, and Joe Neeman.
\newblock Is your function low dimensional?
\newblock In {\em {Conference on Learning Theory, {COLT} 2019}}, volume~99 of
  {\em Proceedings of Machine Learning Research}, pages 979--993, 2019.
\newblock Full version at \texttt{https://arxiv.org/abs/1806.10057}.

\bibitem{de2019junta}
Anindya De, Elchanan Mossel, and Joe Neeman.
\newblock Junta correlation is testable.
\newblock In {\em {2019 IEEE 60th Annual Symposium on Foundations of Computer
  Science (FOCS)}}, pages 1549--1563. IEEE, 2019.

\bibitem{DLM+:07}
I.~Diakonikolas, H.~Lee, K.~Matulef, K.~Onak, R.~Rubinfeld, R.~Servedio, and
  A.~Wan.
\newblock Testing for concise representations.
\newblock In {\em Proc. 48th Ann. Symposium on Computer Science (FOCS)}, pages
  549--558, 2007.

\bibitem{diakonikolas2019distribution}
Ilias Diakonikolas, Themis Gouleakis, and Christos Tzamos.
\newblock {Distribution-independent PAC learning of halfspaces with Massart
  noise}.
\newblock In {\em {Advances in Neural Information Processing Systems}}, pages
  4749--4760, 2019.

\bibitem{diakonikolas2018learning}
Ilias Diakonikolas, Daniel~M Kane, and Alistair Stewart.
\newblock Learning geometric concepts with nasty noise.
\newblock In {\em Proceedings of the 50th Annual ACM SIGACT Symposium on Theory
  of Computing}, pages 1061--1073, 2018.

\bibitem{FKRSS03}
E.~Fischer, G.~Kindler, D.~Ron, S.~Safra, and A.~Samorodnitsky.
\newblock Testing juntas.
\newblock {\em J. Computer \& System Sciences}, 68(4):753--787, 2004.

\bibitem{GR:06}
V.~Guruswami and P.~Raghavendra.
\newblock {Hardness of learning halfspaces with noise}.
\newblock In {\em Proc.\ 47th IEEE Symposium on Foundations of Computer Science
  (FOCS)}, pages 543--552. IEEE Computer Society, 2006.

\bibitem{HalevyKushilevitz:07}
S.~Halevy and E.~Kushilevitz.
\newblock {Distribution-Free Property Testing}.
\newblock {\em SIAM J. Comput.}, 37(4):1107--1138, 2007.

\bibitem{harsha2012invariance}
Prahladh Harsha, Adam Klivans, and Raghu Meka.
\newblock An invariance principle for polytopes.
\newblock {\em {Journal of the ACM (JACM)}}, 59(6):1--25, 2013.

\bibitem{KSS:94}
M.~Kearns, R.~Schapire, and L.~Sellie.
\newblock {Toward Efficient Agnostic Learning}.
\newblock {\em Machine Learning}, 17(2/3):115--141, 1994.

\bibitem{KOS:08}
A.~Klivans, R.~O'Donnell, and R.~Servedio.
\newblock Learning geometric concepts via {G}aussian surface area.
\newblock In {\em Proc.\ 49th IEEE Symposium on Foundations of Computer Science
  (FOCS)}, pages 541--550, 2008.

\bibitem{KNOW14}
P.~Kothari, A.~Nayyeri, R.~O'Donnell, and C.~Wu.
\newblock Testing surface area.
\newblock In {\em Proceedings of the Twenty-Fifth Annual {ACM-SIAM} Symposium
  on Discrete Algorithms, {SODA} 2014, Portland, Oregon, USA, January 5-7,
  2014}, pages 1204--1214, 2014.

\bibitem{Ledoux:94}
M.~Ledoux.
\newblock Semigroup proofs of the isoperimetric inequality in {E}uclidean and
  {G}auss space.
\newblock {\em Bull. Sci. Math.}, 118:485--510, 1994.

\bibitem{Ledoux:00}
Michel Ledoux.
\newblock The geometry of markov diffusion generators.
\newblock In {\em Annales de la Facult{\'e} des sciences de Toulouse:
  Math{\'e}matiques}, volume~9, pages 305--366, 2000.

\bibitem{massart2006risk}
Pascal Massart and {\'E}lodie N{\'e}d{\'e}lec.
\newblock Risk bounds for statistical learning.
\newblock {\em {The Annals of Statistics}}, 34(5):2326--2366, 2006.

\bibitem{MORS:09}
K.~Matulef, R.~O'Donnell, R.~Rubinfeld, and R.~Servedio.
\newblock Testing halfspaces.
\newblock {\em SIAM J.\ Comp.}
\newblock To appear. Extended abstract in Proc.\ Symp.\ Discrete Algorithms
  (SODA) (2009), pp. 256-264. Full version available at
  http://www.cs.cmu.edu/\verb+~+odonnell/.

\bibitem{MORS:10}
K.~Matulef, R.~O'Donnell, R.~Rubinfeld, and R.~Servedio.
\newblock Testing halfspaces.
\newblock {\em SIAM J. on Comput.}, 39(5):2004--2047, 2010.

\bibitem{MORS:09random}
Kevin Matulef, Ryan O'Donnell, Ronitt Rubinfeld, and Rocco~A. Servedio.
\newblock Testing $\pm$1-weight halfspace.
\newblock In {\em {APPROX-RANDOM}}, pages 646--657, 2009.

\bibitem{moore_1955}
E.~H. Moore.
\newblock On the reciprocal of the general algebraic matrix.
\newblock {\em Bulletin of the American Mathematical Society}, 26:394--395,
  1920.

\bibitem{Neeman14}
J.~Neeman.
\newblock Testing surface area with arbitrary accuracy.
\newblock In {\em Symposium on Theory of Computing, {STOC} 2014, New York, NY,
  USA, May 31 - June 03, 2014}, pages 393--397, 2014.

\bibitem{ODonnell:book}
R.~O'Donnell.
\newblock {\em Analysis of Boolean functions}.
\newblock Cambridge University Press, Cambridge, 2014.

\bibitem{parnas2006tolerant}
M.~Parnas, D.~Ron, and R.~Rubinfeld.
\newblock Tolerant property testing and distance approximation.
\newblock {\em {Journal of Computer and System Sciences}}, 72(6):1012--1042,
  2006.

\bibitem{PRS02}
M.~Parnas, D.~Ron, and A.~Samorodnitsky.
\newblock {Testing Basic Boolean Formulae}.
\newblock {\em SIAM J. Disc. Math.}, 16:20--46, 2002.

\bibitem{penrose_1955}
R.~Penrose.
\newblock A generalized inverse for matrices.
\newblock {\em Mathematical Proceedings of the Cambridge Philosophical
  Society}, 51(3):406–413, 1955.

\bibitem{Pisier:86}
G.~Pisier.
\newblock Probabilistic methods in the geometry of {B}anach spaces.
\newblock In {\em Lecture notes in Math.}, pages 167--241. Springer, 1986.

\bibitem{ron2015exponentially}
Dana Ron and Rocco~A Servedio.
\newblock Exponentially improved algorithms and lower bounds for testing signed
  majorities.
\newblock {\em Algorithmica}, 72(2):400--429, 2015.

\bibitem{rudelson2007sampling}
Mark Rudelson and Roman Vershynin.
\newblock Sampling from large matrices: An approach through geometric
  functional analysis.
\newblock {\em {Journal of the ACM (JACM)}}, 54(4):21--es, 2007.

\bibitem{servedio2015adaptivity}
R.~Servedio, L-Y. Tan, and J.~Wright.
\newblock Adaptivity helps for testing juntas.
\newblock In {\em {Proceedings of CCC}}, volume~33. Schloss
  Dagstuhl-Leibniz-Zentrum fuer Informatik, 2015.

\bibitem{vempala2010learning}
Santosh~S Vempala.
\newblock {Learning convex concepts from gaussian distributions with PCA}.
\newblock In {\em {2010 IEEE 51st Annual Symposium on Foundations of Computer
  Science}}, pages 124--130. {IEEE}, 2010.

\end{thebibliography}

\appendix
\section{Perturbation bounds for subspaces and smooth functions}

\begin{lemma}\label{lem:distance-between-averages}
    If $E$ and $E'$ are two subspaces of $\mathbb{R}^n$ then for any Lipschitz function $f: \mathbb{R}^n \to \mathbb{R}$,
    \[
        \E |\Avg_E f(\bx) - \Avg_{E'} f(\bx)|
        \le \sqrt 2 c \|\Pi_E - \Pi_{E'}\|_F.
    \]
    Here $c$ is the Lipschitz constant of $f$. 
\end{lemma}

\begin{proof}
    Since we can express $\Avg_{E'} f(x)$ as
    \[
        \Avg_{E'} f(x) = \E [f(\Pi_E x + \Pi_{E^\perp} \bz + \left(\Pi_{E'} - \Pi_E\right)x + (\Pi_{E^\perp} - \Pi_{(E')^\perp}) \bz)],
    \]
    the Lipschitz property allows us to bound
    \[
        |\Avg_E f(x) - \Avg_{E'} f(x)| \le c\E \left[\|\left(\Pi_{E'} - \Pi_E\right)x + (\Pi_{E^\perp} - \Pi_{(E')^\perp}) \bz\|\right].
    \]
    Since $\Pi_{E^\perp} = I - \Pi_{E}$ (and similarly for $E'$),
    the right hand side is just
    $
        c \E \|(\Pi_{E'} - \Pi_E) (x - \bz)\|.
    $
Observe that if $\bx$ and $\bz$ are distributed as a standard Gaussian, then $\bx-\bz$ is distributed as $N(0,\sqrt{2})$.     
     Hence,
    \[
        \E_{\bx} |\Avg_E f(\bx) - \Avg_{E'} f(\bx)| \le \sqrt 2 c
        \E_{\bz} \left\|\left(\Pi_{E'} - \Pi_E\right) \bz\right\|.
    \]
    Finally, one can easily check that for any matrix $A$, $\E\|A\bz\|^2 = \|A\|_F^2$;
    by Jensen's inequality, it follows that $\E\|AZ\| \le \|A\|_F$ which finishes the proof. 
\end{proof}



\begin{lemma}\label{lem:subspace-distance}
    If $E$ and $E'$ are two subspaces of $\mathbb{R}^n$ such that $\|\Pi_E \Pi_{(E')^\perp}\| < 1$, 
    then there exists a subspace $\tilde E \subset E'$ with $\dim \tilde E = \dim E$ such that
    \[
        \|\Pi_E - \Pi_{\tilde E}\|_F^2 \le 8 \|\Pi_E \Pi_{(E')^\perp}\|_F^2.
    \]
\end{lemma}

\begin{proof}
Let $k = \dim E$ and  define $A = \Pi_E \Pi_{E'}$.  
We begin with the following simple claim. 
\begin{claim}
The matrix $A$ has exactly $k$ non-zero singular values. 
\end{claim}
\begin{proof}
Observe that it suffices to analyze the eigenvalues of $A\cdot A^T$. Towards this, we observe that for any $w \in E^{\perp}$, $w^T \cdot A \cdot A^T \cdot w =0$. On the other hand, for any $w \in E$, 
\begin{align*}
w^T \cdot A \cdot A^T \cdot w = w^T \Pi_E \Pi_{E'} \Pi_{E'} \Pi_E \cdot w  = w^T \cdot  \Pi_E \cdot\Pi_{E'}  \cdot w \ge w^T \cdot \Pi_E \cdot w -  w^T \cdot \Pi_E \cdot \Pi_{(E')^\perp} \cdot w >0. 
\end{align*}
The last inequality uses that $w \in E$ and $\Vert \Pi_E \Pi_{(E')^\perp} \Vert_2 <1$. 
This implies that for all $w \in E$, $w^T \cdot A\cdot A^T \cdot w >0$. Consequently, $A \cdot A^T$ (and hence $A$) has exactly $k$ non-zero eigenvalues.
\end{proof}
Let $U D V^T = \Pi_E \Pi_{E'} = A$ be the singular value decomposition  of $A$  and 
observe that we can assume that $D$ is $k \times k$ diagonal matrix (whose diagonal elements are positive).   Let the columns of $U$ be $\{u_1, \ldots, u_k\}$ and the columns of $V$ be $\{v_1, \ldots, v_k\}$. Then, $\{u_1, \ldots, u_k\}$ is an orthonormal basis for $E$ and and 
$\{v_1, \ldots, v_k\}$ is an orthonormal basis for a subspace $\tilde{E}$ of $E'$. Observe that   $\Pi_E = U U^T$ and $\Pi_{\tilde E} = V V^T$. Thus, 
    \begin{align*}
        \|\Pi_E - \Pi_E \Pi_{E'}\|_F^2 &= \|U U^T - U D V^T\|_F^2 \\
        &= \|U^T - D V^T\|_F^2 \\
        &= \sum_{i=1}^k \|d_i v_i - u_i\|_2^2.
    \end{align*}
    For fixed unit vectors $v$ and $u$, $\|d v - u\|_2^2$ is minimized when $d = \inr{u}{v}$,
    in which case $\|d v - u\|_2^2 = 1 - \inr{u}{v}^2$. If $d \ge 0$,
    $\|d v - u\|_2^2 \ge \frac 12 \|v - u\|_2^2$.
    Hence,
    \begin{equation}~\label{eq:pi-diff}
        \|\Pi_E - \Pi_E \Pi_{E'}\|_F^2 = \|U D V^T - U U^T\|_F^2 \ge \frac 12 \sum_{i=1}^k \|u_i - v_i\|_2^2.
    \end{equation}
    Finally,
    \begin{align*}
        \|\Pi_E - \Pi_{\tilde E}\|_F^2
        &= \|U U^T - V V^T\|_F^2 \\
        &\le 2\|U U^T - U V^T\|_F^2 + 2 \|U V^T - V V^T\|_F^2 \\
        &= 2\|U^T - V^T\|_F^2 + 2 \|U - V\|_F^2 \\
        &= 4 \sum_{i=1}^k \|u_i - v_i\|_2^2.
    \end{align*}
    Combining this with \eqref{eq:pi-diff},  we have $\|\Pi_E - \Pi_{\tilde E}\|_F^2 \le 8 \|\Pi_E - \Pi_E \Pi_{E'}\|_F^2$.
\end{proof}

\begin{corollary}\label{cor:correlation-subspace-stability}
Let $f: \mathbb{R}^n \to [-1, 1]$ be a $c$-Lipschitz function. Let $E'$ be another subspace of $\mathbb{R}^n$ such that $\|\Pi_E \Pi_{(E')^\perp}\| < 1$. Then, there exists a subspace $\tilde{E}$ of $E'$ such that
for any $g: \mathbb{R}^n \to [-1, 1]$,  
    \[
        \big| \mathbf{E}_{\bx} [f(\bx) \calA_E g(\bx)] -  \mathbf{E}_{\bx} [f(\bx)  \calA_{\tilde{E}} g(\bx)] \big| \le 4 c\|\Pi_E \Pi_{(E')^\perp}\|_F.
    \]
\end{corollary}

\begin{proof}
Let $\epsilon = \|\Pi_E \Pi_{(E')^\perp}\|_F$ (by assumption $\epsilon<1$). Then, using Lemma~\ref{lem:subspace-distance}, it follows that  there exists a subspace $\tilde{E}$  of $E'$ with $\mathsf{dim}(\tilde{E}) = \mathsf{dim}(E)$  such that 
$\Vert \Pi_E-\Pi_{\tilde{E}} \Vert_F \le 2\sqrt{2} \epsilon$. Since $f$ is $c$-Lipschitz, it follows that 
\begin{equation}
\mathbf{E}_{\bx}[|\Avg_E f(\bx) - \Avg_{\tilde{E}} f(\bx)|] \le \sqrt{2} c \Vert \Pi_E - \Pi_{\tilde{E}} \Vert_F  \le 4 c \Vert \Pi_E \Pi_{E'^{\perp}} \Vert_F. 
\end{equation}
 The first inequality uses Lemma~\ref{lem:distance-between-averages} and the second inequality uses Lemma~\ref{lem:subspace-distance}. Thus, it follows that 
 \begin{equation}~\label{eq:diff-correlation-ae}
 \big|\mathbf{E}_{\bx}[g(\bx) \Avg_E f(\bx)] - \mathbf{E}_{\bx}[g(\bx)  \Avg_{\tilde{E}} f(\bx)]\big| \leq \mathbf{E}_{\bx}[|g(\bx)| \cdot | \Avg_E f(\bx) - \Avg_{\tilde{E}} f(\bx)|] \le 4 c \Vert \Pi_E \Pi_{E'^{\perp}} \Vert_F.  
 \end{equation}
 Using the fact that $g: \mathbb{R}^n \rightarrow [-1,1]$, the claim follows. 

\end{proof}
\section{Stability of the pseudoinverse}
In this section, we prove a stability result 
for pseudoinverse of square matrices. 
We first recall the well-known Davis-Kahan theorem -- the precise formulation we are using is Theorem~VII.~3.~2 
from \cite{bhatia2013matrix}. 
\begin{theorem}~\label{thm:DavisKahan}
Let $A, B \in \mathbb{C}^{n \times n}$ be Hermitian matrices. Let $S_1 = [a,b]$ and $S_2 = \mathbb{R} \setminus [a-\delta, b+\delta]$ for $\delta>0$. Let $E_1$ be the subspace spanned by the eigenvectors of $A$ corresponding to eigenvalues lying in $S_1$. Similarly, $E_2$ is the subspace spanned by the eigenvectors of $B$ corresponding to eigenvalues lying in $S_2$. Then, for any unitarily invariant norm $\Vert \cdot \Vert$ (such as the $\Vert \cdot \Vert_2$ or $\Vert \cdot \Vert_F$), 
\[
\| \Pi_{E_1} \Pi_{E_2} \| \le \frac{\pi}{2\delta} \Vert A-B \Vert. 
\]
\end{theorem}
Using the Davis-Kahan theorem, we prove the following result. Before we do that, we use the following simple fact. 
\begin{fact}~\label{fact:almost-same-eigen}
Let $R$ be a psd matrix and $W$ be the space spanned by eigenvectors with eigenvalues in the range $[\lambda-\delta, \lambda + \delta]$ where $\lambda>2\delta$. Then, for any $w \in W$ 
\[
\big\Vert R^{-1} w - \frac{1}{\lambda} w \big\Vert_2 \le \frac{2\delta}{\lambda^2} \Vert w \Vert_2. 
\]
\end{fact}
\begin{proof}
Express $w = \sum_i \alpha_i w_i$ (where $\{w_i\}$ are orthonormal eigenvectors of $R$). If $\lambda_i$ is the eigenvalue corresponding to $w_i$, then note that 
\[
R^{-1} w= \sum_i \frac{\alpha_i}{\lambda_i} w_i. 
\]
This implies that 
\[
R^{-1} w - \frac{w}{\lambda} = \sum_i {\alpha_i} w_i \cdot \big( \frac{1}{\lambda_i} - \frac{1}{\lambda} \big)
\]
From this, it follows that 
\[
\big\Vert R^{-1} w - \frac{1}{\lambda} w \big\Vert_2 \le \frac{2\delta}{\lambda^2} \Vert w \Vert_2. 
\]
\end{proof}

\begin{lemma}\label{lem:pseudoinverse-stab}
Let $A, \tilde{A} \in \mathbb{R}^{n \times n}$ be psd matrices. Let $\eta \ge 0$ and $V$ denote the subspace spanned by the eigenvalues of $A$ in $[\eta, \infty)$. Then, 
\[
\Vert (A_{\ge \eta}^{-1} - \tilde{A}_{\ge \eta/2}^{-1}) \cdot \Pi_V \Vert \le \frac{20\Vert A \Vert_F \sqrt{ \Vert A- \tilde{A} \Vert_2}}{\eta^{5/2}}. 
\]
\end{lemma}
\begin{proof}
Define $B =  (A_{\ge \eta}^{-1} - \tilde{A}_{\ge \eta/2}^{-1}) \cdot \Pi_V$. 
First of all, observe that for all $v \in V^{\perp}$, $Bv=0$. Next, consider any eigenvector $v$ of $A$ with corresponding eigenvalue $\lambda \in [\eta, \infty)$. 
For a choice of $\delta>0$ to be fixed later, define 
$W$ to be the span of the eigenvectors of $\tilde{A}$
with eigenvalues in $ [\lambda - \delta, \lambda+\delta]$. Then, Theorem~\ref{thm:DavisKahan} implies that $v_{W^\perp}$ (i.e., the projection of $v$ on $W^\perp$) satisfies 
\begin{equation}\label{eq:small-bound2}
\Vert v_{W^\perp} \Vert_2 \le \frac{\pi}{2\delta} \Vert A -\tilde{A} \Vert_2. 
 \end{equation} 
 Consequently,  as the largest singular value of $\tilde{A}_{\ge \eta/2}^{-1}$ is at most $2/\eta$, 
 \begin{equation}\label{fact:small-bound1}
 \Vert \tilde{A}_{\ge \eta/2}^{-1} v_{W^\perp} \Vert_2 \le \frac{\pi}{\eta \cdot \delta} \Vert A -\tilde{A} \Vert_2. 
 \end{equation}

On the other hand,  using Fact~\ref{fact:almost-same-eigen}

\begin{eqnarray}
\Vert (A_{\ge \eta}^{-1} - \tilde{A}_{\ge \eta/2}^{-1}) \cdot \Pi_V v\Vert_2 &=&\Vert (A_{\ge \eta}^{-1} - \tilde{A}_{\ge \eta/2}^{-1})  v\Vert_2 \nonumber \\
&\le & \big\Vert \frac{v}{\lambda}  - \tilde{A}_{\ge \eta/2}^{-1} v_W \big\Vert  + \Vert \tilde{A}_{\ge \eta/2}^{-1} v_{W^{\perp}} \Vert \nonumber \\ 
&\le& \big\Vert \frac{v}{\lambda}  - \tilde{A}_{\ge \eta/2}^{-1} v_W \big\Vert + \frac{\pi}{\eta \cdot \delta} \Vert A -\tilde{A} \Vert_2 \ \ \textrm{using \eqref{fact:small-bound1}} \nonumber \\
&\le& \big\Vert \frac{v}{\lambda}  - \frac{ v_W}{\lambda} \big\Vert + \frac{2\delta}{\lambda^2} \Vert v_W \Vert_2 + \frac{\pi}{\eta \cdot \delta} \Vert A -\tilde{A} \Vert_2  \ \ \textrm{using Fact~\ref{fact:almost-same-eigen}} \nonumber \\ 
&=& \frac{ \Vert v_{W^\perp} \Vert}{\lambda} + \frac{2\delta}{\lambda^2} \Vert v_W \Vert_2 + \frac{\pi}{\eta \cdot \delta} \Vert A -\tilde{A} \Vert_2  \nonumber \\
&\le&  \frac{\pi}{\lambda \cdot \delta} \Vert A -\tilde{A} \Vert_2 + \frac{\pi}{\eta \cdot \delta} \Vert A -\tilde{A} \Vert_2 + \frac{2\delta}{\lambda^2}. \nonumber
\end{eqnarray} 
The last inequality applies \eqref{eq:small-bound2} and the fact that $v$ is a unit vector.  Next, we observe that $\lambda \ge \eta/2$ and thus, the above  inequality implies 
\[
\Vert (A_{\ge \eta}^{-1} - \tilde{A}_{\ge \eta/2}^{-1}) \cdot \Pi_V v\Vert_2  \leq \frac{2\pi}{\eta \cdot \delta} \Vert A -\tilde{A} \Vert_2 + \frac{\pi}{\eta \cdot \delta} \Vert A -\tilde{A} \Vert_2 + \frac{8\delta}{\eta^2}.
\]
By an appropriate choice of $\delta =  \sqrt{\eta \Vert A- \tilde{A} \Vert}$, we have 
\[
\Vert (A_{\ge \eta}^{-1} - \tilde{A}_{\ge \eta/2}^{-1}) \cdot \Pi_V v\Vert_2  \le 20 \frac{\sqrt{\Vert A- \tilde{A} \Vert_2}}{\eta^{3/2}}. 
\] 
Suppose $u$ is any  unit vector in $V$. Then, it follows from the above equation that 
\[
\Vert (A_{\ge \eta}^{-1} - \tilde{A}_{\ge \eta/2}^{-1}) \cdot \Pi_V v\Vert_2  \le20 \frac{\sqrt{\mathsf{dim}(V)}\sqrt{\Vert A- \tilde{A} \Vert_2}}{\eta^{3/2}}.
\]
Finally, observe that $\sqrt{\mathsf{dim}(V)} \le \frac{\Vert A \Vert_F}{\eta}$.  This finishes the proof.


\end{proof}
\end{document}